\documentclass[final]{dmtcs-episciences}

\usepackage[utf8]{inputenc}
\usepackage{subfigure}

\usepackage{amsmath,amsthm,enumerate,amssymb}
\usepackage{hyperref}
\usepackage{pgf,tikz}
\usetikzlibrary{arrows,calc}
\usetikzlibrary{positioning}
\definecolor{cqcqcq}{rgb}{0.75,0.75,0.75}
\usepackage{comment}

\newcommand{\np}{$\mathcal{NP}$}
\newcommand{\p}{$\mathcal{P}$}
\newcommand{\W}{{\rm W}}
\newcommand{\I}{{\rm I}}

\newcommand{\qedclaim}{\hfill $\diamond$ \medskip}
\newenvironment{proofclaim}{\noindent{\em Proof of the claim.}}{\qedclaim}

\newtheorem{theorem}{Theorem}[section]

\newtheorem{lemma}[theorem]{Lemma}
\newtheorem{observation}[theorem]{Observation}

\newtheorem{corollary}[theorem]{Corollary}
\newtheorem{definition}[theorem]{Definition}

\newcounter{claimOC}[theorem]

\newenvironment{claimOC}[1][]{\refstepcounter{claimOC}\par\medskip
   \noindent \textbf{Claim~\theclaimOC. #1} \rmfamily}{\medskip}

\usetikzlibrary{decorations.pathreplacing}

\makeatother

\newcommand{\defproblem}[3]{
    \vspace{3mm}
    \noindent\fbox{
        \begin{minipage}{0.95\textwidth}
            #1\newline
            \textbf{Input:} #2\\
            \textbf{Task:} #3
        \end{minipage}
    }
    \vspace{3mm}
}

\usepackage[round]{natbib}

\author[F. Fioravantes, N. Melissinos and T. Triommatis]{Foivos Fioravantes\affiliationmark{1}\thanks{International Mobility of Researchers MSCA-F-CZ-III at CTU in Prague, $\text{CZ}.02.01.01/00/22\_010/0008601$ Programme Johannes Amos Comenius..}
  \and Nikolaos Melissinos\affiliationmark{1}\thanks{European Union under the project Robotics and advanced industrial production (reg. no. CZ.02.01.01/00/22\_008/0004590).}
  \and Theofilos Triommatis\affiliationmark{2}}
\title[Parameterised distance to local irregularity]{Parameterised distance to local irregularity}

\affiliation{Department of Theoretical Computer Science, Faculty of Information Technology, Czech Technical University in Prague, Prague, Czech Republic
\\School of Production Engineering and Management, Technical University of Crete, Chania, Greece
}

\begin{document}

\publicationdata{vol. 27:3}{2025}{28}{10.46298/dmtcs.15004}{2024-12-27; 2024-12-27; 2025-10-29}{2025-11-27}
\maketitle

\begin{abstract}
  A graph $G$ is \emph{locally irregular} if no two of its adjacent vertices have the same degree. The authors of [Fioravantes et al. Complexity of finding maximum locally irregular induced subgraph. {\it SWAT}, 2022] introduced and provided some initial algorithmic results on the problem of finding a locally irregular induced subgraph of a given graph $G$ of maximum order, or, equivalently, computing a subset $S$ of $V(G)$ of minimum order, whose deletion from $G$ results in a locally irregular graph; $S$ is called an \emph{optimal vertex-irregulator of $G$}. In this work we provide an in-depth analysis of the parameterised complexity of computing an optimal vertex-irregulator of a given graph $G$. Moreover, we introduce and study a variation of this problem, where $S$ is a subset of the edges of $G$; in this case, $S$ is denoted as an \emph{optimal edge-irregulator of $G$}. We prove that computing an optimal vertex-irregulator of a graph $G$ is in FPT when parameterised by various structural parameters of $G$, while it is $\W[1]$-hard when parameterised by the feedback vertex set number or the treedepth of $G$. Moreover, computing an optimal edge-irregulator of a graph $G$ is in FPT when parameterised by the vertex integrity of $G$, while it is \np-hard even if $G$ is a planar bipartite graph of maximum degree $6$, and $\W[1]$-hard when parameterised by the size of the solution, the feedback vertex set or the treedepth of $G$. Our results paint a comprehensive picture of the tractability of both problems studied here.
\end{abstract}

\section{Introduction}

A fundamental problem in graph theory is ``given a graph $G$, find an induced subgraph $H$ of $G$, of maximum order, that belongs in the family of graphs verifying a property $\Pi$'', in which case we say that $H\in \Pi$:

\defproblem{\textsc{\textsc{Largest Induced Subgraph with Property $\Pi$} (ISP-$\Pi$)}\cite{GJ79}}{
    A graph $G=(V,E)$, an integer $k$, a property $\Pi$.
}{
    Does there exist a set $S\subseteq V$ such that $|S|\leq k$ and $G-S\in \Pi$?
}

\noindent There is a plethora of classical problems that fall under this general setting. Consider, for example, the \textsc{Vertex Cover} and the \textsc{Feedback Vertex Set}, where $\Pi$ is the property ``the graph is an independent set'' and ``the graph is a forest'', respectively.  

In this paper we study the ISP-$\Pi$ problem where $\Pi$ is the property ``the graph is locally irregular'', recently introduced in~\cite{FMT22}. A graph $G=(V,E)$ is called \emph{locally irregular} if no two adjacent vertices in $V$ have the same degree. We extend the work of~\cite{FMT22}, by more thoroughly investigating the parameterised behaviour of the problem. In addition, we take the first step towards the problem of finding large locally irregular (not necessarily induced) subgraphs of a given graph $G$. In particular, we introduce the problem where the goal is to find a subset of edges of $G$ of minimum order, whose removal renders the graph locally irregular. Our results allow us to paint a rather clear picture concerning the tractability of both problems studied here in relation to various standard graph-structural parameters (see Figure~\ref{fig:parameters} for an overview of our results).

\begin{figure}[htb]
    \centering
    
    \begin{tikzpicture}[every node/.style={thick, align=center}]
  \small

  \node (vc) at (0,   0) {vertex cover};
  \node (nd) at (-4, 1) {\textbf{neighborhood} \\ \textbf{diversity}};
  \node[above = 6 of vc] (cw) {clique-width};

  \node[above = .5 of vc] (tco)  {twin cover $\operatorname{+} \omega$};
  \node[above = .8 of tco] (anccd) {};

  \node[left =  .2 of anccd] (tc)  {twin cover};
  \node[right = .1 of anccd] (cdo) {cluster deletion $\operatorname{+} \omega$};
  \node[above = .5 of cdo.east] (vi)  {\textbf{vertex integrity}};

  \node[above = .5 of vi] (td)  {\textbf{treedepth*}};

  \node[above = 2 of nd] (mw) {modular-width};
  \node[above = .5 of td] (pw) {pathwidth};

  \node[above = .5 of pw] (tw) {treewidth};

  \node[right = 1.5 of td.south] (fvs) {\textbf{feedback vertex set*}};

  \node[above = of anccd] (cd) {\textbf{cluster deletion}};
  \node[above = of cd] (sd) {shrub-depth};

  \draw[thick] (cw) -- (mw) -- (nd) -- (vc);
  \draw[thick] (tw) -- (pw) -- (td);
  \draw[thick] (vi) -- (cdo);
  \draw[thick] (tw) -- (fvs) -- (vc);

  \draw[thick] (cw) -- (sd) -- (cd);
  \draw[thick] (sd) -- (td);
  \draw[thick] (sd) -- (nd);

  \draw[thick] (cd) -- (cdo) -- (tco) -- (vc) ;
  \draw[thick] (tc) -- (tco);
  \draw[thick] (cw) -- (tw);

  \draw[thick] (td) -- (vi);
  \draw[thick] (nd) -- (vc);
  
  \draw[thick] (cd) -- (tc);

  \draw[thick] (tc) -- (mw);



  \coordinate (endnph) at ($(cw)+(-2, 0)$);
  \coordinate (fvs-vc) at ($(fvs)!0.1!(vc)$);
  \coordinate (td-vi) at ($(td)!0.35!(vi)$);
  \coordinate (whard) at ($(fvs-vc) + (2, 0)$);
  \coordinate (sd-cd) at ($(sd)!0.6!(cd)$);
  \coordinate (endwhard) at ($(endnph)+(-1.25, -0.5)$);
  
  \draw[thick, red] (whard) to [out=180, in=0] (fvs-vc);
  \draw[thick, red] (fvs-vc) to [out=180, in=0] (td-vi);
  \draw[thick, red] (td-vi) to [out=180, in=0] ($(sd-cd) + (0,0.2)$);
  \draw[thick, red] ($(sd-cd) + (0,0.2)$) to [out=180, in=-40] (endwhard) node[above] {$\uparrow$W[1]-h};

  \coordinate (nd-mw) at ($(nd)!0.4!(mw)$);
  \coordinate (kfpt) at ($(nd-mw) + (-1,0)$);
  \coordinate (endfpt) at ($(fvs) + (0.8, -0.7)$);

  \draw[thick, cyan] (kfpt) node[below] {FPT $\I_v \downarrow$} -- (nd-mw);
  \draw[thick, cyan] (nd-mw) to [out=0, in=180] (sd-cd)[out=0,in=180];
  \draw[thick, cyan] (sd-cd) to [out=0, in=180] ($(td-vi) + (0,-0.2)$)[out=0,in=180];
  \draw[thick, cyan] ($(td-vi) + (0,-0.2)$) to [out=0, in=180] ($(fvs-vc) + (0,-0.2)$)[out=0,in=180];
  \draw[thick, cyan] ($(fvs-vc) + (0,-0.2)$) to [out=0, in=180] ($(whard) + (0,-0.2)$);

  \coordinate (vc-nd) at ($(vc)!0.5!(nd)$);
  \coordinate (tc-tcw) at ($(tc)!0.4!(tco)$);
  \coordinate (cd-cdo) at ($(cd)!0.5!(cdo)$);

  \draw[thick, olive] ($(vc-nd) + (-1,-0.5)$) to [out=50, in=210] (tc-tcw);
  \draw[thick, olive] (tc-tcw) to [out=30, in=210] (cd-cdo);
  \draw[thick, olive] (cd-cdo) to [out=30, in=180] ($(td-vi) + (0,-0.38)$);
  \draw[thick, olive] ($(td-vi) + (0,-0.38)$) to [out=0, in=180] ($(fvs-vc) + (0,-0.4)$);
  \draw[thick, olive] ($(fvs-vc) + (0,-0.4)$) to [out=0, in=180] ($(whard) + (0,-0.4)$) node[below left] {$\downarrow$ FPT $\I_e$};

  \draw[thick, cyan] (kfpt) node[below] {FPT $\I_v \downarrow$} -- (nd-mw);

  \coordinate (vifpt) at ($(vi.east) + (0, -1)$);
  \coordinate (cdfpt) at ($(cd.east) + (1, 0)$);



  \coordinate (aboutmw) at ($(mw) + (0,1)$);


  %

  \end{tikzpicture}
    \caption{Overview of our results. 
    A parameter $p$ appearing linked to a parameter $p'$ with $p$ being below $p'$ is to be understood as ``there is a function $f$ such that, for each graph $G$, we have $p(G)\le f(p'(G))$''. The bold font is used to indicate the parameters that we consider in this work. The asterisks are used to indicate that the corresponding result follows from observations based on the work in~\cite{FMT22}. In light blue (olive resp.) we exhibit the FPT results we provide for finding an optimal vertex (edge resp.) irregulator, denoted as $\I_v$ ($\I_e$ resp.). In red we exhibit the $\W[1]$-hardness results we provide for both problems. The clique number of the graph is denoted by $\omega$.}
    \label{fig:parameters}
\end{figure}

\noindent\textbf{ISP-$\boldsymbol{\Pi}$ and hereditarity.} The ISP-$\Pi$ problem has been extensively studied for \emph{hereditary} properties. A property $\Pi$ is hereditary if, for any graph $G$ verifying it, any induced subgraph of $G$ also verifies that property. The properties ``the graph is an independent set'' or ``the graph is a forest'' are, for example, hereditary. It was already shown in~\cite{LY80} that ISP-$\Pi$ is an \np-hard problem for any non-trivial hereditary property. On the positive side, the ISP-$\Pi$ problem always admits an FPT algorithm, when parameterised by the size of the solution, if $\Pi$ is a hereditary property~\cite{c96,KR02}. This is an important result, as it allows us to conceive efficient algorithms to solve computationally hard problems, as long as we restrict ourselves to graphs verifying such properties. 

It is also worth mentioning the work in~\cite{FGLS19}, which provides a framework that yields exact algorithms that are significantly faster than brute-force to solve a more general version of the ISP-$\Pi$ problem: given a universe, find a subset of maximum cardinality which verifies some hereditary property. On a high level, the algorithm proposed in~\cite{FGLS19} builds the solution which is a subset $H$ of maximum cardinality with the wanted property, by continuously extending a partial solution $X\subseteq H$. Note that this approach only works if $\Pi$ is indeed a hereditary property. 
More recently, this approach was generalised by the authors of~\cite{EKMNS22}, who provide a framework that yields exponential-time approximation algorithms. 

However, not all interesting properties are hereditary. E.g., ``all vertices of the induced subgraph have odd degree'', and ``the induced subgraph is $d$-regular'', where $d$ is an integer given in the input (recall that a graph is $d$-\emph{regular} if all of its vertices have the same degree $d$), are non-hereditary properties. The authors of~\cite{BS21} studied the ISP-$\Pi$ problem for the former property, showing that it is an \np-hard problem, and providing an FPT algorithm that solves it when parameterised by the rank-width. 
Also, the authors of~\cite{AGTHP17,AEIM14,MT09} studied the ISP-$\Pi$ problem for the latter property. It is shown in~\cite{AEIM14} that finding an induced subgraph of maximum order that is $d$-regular is \np-hard to approximate, even on bipartite or planar graphs. The authors of~\cite{AEIM14} also provide a linear-time algorithm to solve this problem for graphs with bounded treewidth. Lastly, it is also worth mentioning~\cite{VHNW13}, where the authors consider the non-hereditary property ``the induced subgraph is $k$-anonymous'', where a graph $G$ is $k$-anonymous if for each vertex of $G$ there are at least $k-1$ other vertices of the same degree. 

An important observation is that, in the case of non-hereditary properties, the ISP-$\Pi$ problem does not necessarily admit an FPT algorithm parameterised by the size of the solution.
Indeed, the authors of~\cite{MT09} proved that when considering $\Pi$ as ``the induced subgraph is regular'', the ISP-$\Pi$ problem is $\W[1]$-hard when parameterised by the size of the solution. 
This indicates the importance of considering graph-structural parameters for conceiving efficient algorithms for such problems. This is exactly the approach followed in~\cite{GKO21,LV23}, where the authors consider a generalisation of \textsc{Vertex Cover}, the ISP-$\Pi$ problem where $\Pi$ is ``the graph has maximum degree $k$'', for an integer $k$ given in the input.

\medskip

\noindent\textbf{Distance from local irregularity.} In some sense, the property that interests us lies on the opposite side of the one studied in~\cite{AGTHP17,AEIM14,MT09}. Recall that a graph $G$ is locally irregular if no two of its adjacent vertices have the same degrees. This notion was formally introduced in~\cite{BBPW15}, where the authors take some steps towards proving the so-called 1-2-3 Conjecture proposed in~\cite{KLT04} and recently proven in~\cite{K24}. Roughly, this conjecture is about functions assigning weights from $[k]=\{1,\dots,k\}$ to the edges of a graph, called proper $k$-labellings, so that all adjacent vertices have different weighted degrees; the conjecture states that for any non-trivial graph, this is always achievable for $k\leq 3$. 

The authors of~\cite{FMT22} introduced the problem of finding a locally irregular induced subgraph of a given graph $G$ of maximum order (a non-hereditary property). Equivalently, find a set of \emph{vertices} of minimum cardinality whose deletion renders the given graph locally irregular; such sets are named \emph{optimal vertex-irregulators}. The main focus of~\cite{FMT22} was to study the complexity of computing optimal vertex-irregulators. It was shown that this problem is \np-hard even for subcubic planar bipartite graphs, $\W[2]$-hard parameterised by the size of the solution and $\W[1]$-hard parameterised by the treewidth of the input graph. Moreover, for any constant $\varepsilon <1$, there cannot be a polynomial-time $\mathcal{O}(n^{1-\varepsilon})$-approximation algorithm (unless \p$=$\np). On the positive side, there are two FPT algorithms that solve this problem, parameterised by the maximum degree of the input graph plus either the size of the solution or the treewidth of the input graph. Note that the notion of vertex-irregulators proved to be fruitful in the context of proper labellings. Indeed, the author of~\cite{B25} observed a connection between finding large locally irregular induced subgraphs and constructing proper $k$-labellings that also maximise the use of weight $1$ on the edges of the given graph.

Apart from improving the results of~\cite{FMT22}, in this paper we also introduce the novel problem of computing a subset of a graph's \emph{edges}, of minimum order, whose deletion renders the graph locally irregular; such sets are named \emph{optimal edge-irregulators}. This problem is introduced as a first step towards understanding the problem of finding large locally irregular (not necessarily induced) subgraphs of a given graph. Problems concerned with finding maximum subgraphs verifying a specific property have also been extensively studied (\emph{e.g.},~\cite{CFKZ12,CHW19,AKS23}).
One might expect that finding edge-irregulators could be easier than finding vertex-irregulators as it is often the case with graph theoretical problems concerned with subsets of edges, whose versions considering subsets of vertices are intractable (recall, \emph{e.g.}, \textsc{Edge Cover}, \textsc{Feedback Edge Set} and even  \textsc{Min Weighted Lower-Upper-Cover}~\cite{S03}). 
This expectation might also be enforced by the fact that there is no general result like the one in~\cite{LY80} showing the computational hardness for the problem of deleting edges to reach a graph verifying a specific property (hereditary or otherwise). This latter problem has been studied for various properties, e.g., the ones of being $P_k$-free~\cite{MC88}, $C_k$-free~\cite{Y81}, triangle-free~\cite{ASS05}. It turns out that for many properties, this problem remains \np-hard. In particular, the authors of~\cite{KPS24} provide a framework to show the hardness of this problem for some properties when the input graph belongs to specific families of graphs, although this framework is not as general as the one in~\cite{LY80}. In this work, we establish that finding small edge-irregulators is also a computationally hard problem. 

\medskip

\noindent\textbf{Our contribution.} We study the complexity of computing optimal vertex and edge-irregulators. We identify the parameters for which the tractability of the former problem changes, considering a multitude of standard graph-structural parameters. We also take steps towards the same goal for the latter problem. In Section~\ref{sec:prelim} we introduce the needed notation and provide some first results. In particular, we observe that computing optimal vertex-irregulators is $\W[1]$-hard when parameterised by the treedepth or the feedback vertex set of the given graph. Section~\ref{sec:FPT-general} provides FPT algorithms for the problem of finding optimal vertex-irregulators. The considered parameters are the neighborhood diversity, the vertex integrity, or the cluster deletion number of the input graph. In Section~\ref{sec:hardness}, we focus on the problem of finding optimal edge-irregulators. First, we prove that this problem is \np-hard, even when restricted to planar bipartite graphs of maximum degree $6$. We also show that the problem is $\W[1]$-hard parameterised by the size of the solution or the feedback vertex set of the input graph. Lastly, we modify the FPT algorithm for computing an optimal vertex-irregulator parameterised by the vertex integrity in order to provide an FPT algorithm that solves the edge version of the problem (once more parameterised by the vertex integrity). We close the paper in Section~\ref{sec:conclusion}, where we propose some directions for further research. 

\section{Preliminaries}\label{sec:prelim}
We follow standard graph theory notations~\cite{D12}. 

Let $G=(V,E)$ be a graph and $G'=(V',E')$ be a subgraph of $G$ ({\it i.e.}, created by deleting vertices and/or edges of $G$). Recall first that the subgraph $G'$ is \emph{induced} if it can be created only by deleting vertices of $G$; in this case we denote $G'$ by $G[V']$. That is, for each edge $uv\in E$, if $u,v\in V'$, then $uv\in E'$. For any vertex $v\in V$, let $N_G(v)=\{u\in V : uv\in E\}$ denote the \emph{neighbourhood} of $v$ in $G$ 
and $d_G(v)=|N_G(v)|$ denote the \emph{degree} of $v$ in $G$. Note that, whenever the graph $G$ is clear from the context, we will omit the subscript and simply write $N(v)$ and $d(v)$. 
Also, for $S\subseteq E$, we denote by $G-S$ the graph $G'=(V, E\setminus S)$. That is, $G'$ is the graph resulting from the deletion of the edges of $S$ from the graph $G$.

Let $G=(V,E)$ be a graph. We say that $G$ is \emph{locally irregular} if, for every edge $uv\in E$, we have $d(u)\neq d(v)$. Now, let $S\subseteq V$ be such that $G[V\setminus S]$ is a locally irregular graph; any set $S$ that has this property is denoted as a \emph{vertex-irregulator of $G$}. Moreover, let $\I_v(G)$ be the minimum order that any vertex-irregulator of $G$ can have. We will say that $S$ is an \emph{optimal} vertex-irregulator of $G$ if $S$ is a vertex-irregulator of $G$ and $|S|=\I_v(G)$. Similarly, we define an 
\emph{edge-irregulator of $G$} to be any set $S\subseteq E$ such that $G-S$ is locally irregular. Moreover, let $\I_e(G)$ be the minimum order that any edge-irregulator of $G$ can have. We will say that $S$ is an \emph{optimal} edge-irregulator of $G$ if $S$ is an edge-irregulator of $G$ and $|S|=\I_e(G)$.

We begin with some simple observations that hold for any graph $G=(V,E)$. 
\begin{observation}\label{obs:neigbourhood}
If $G$ contains two vertices $u,v$ such that $uv\in E$ and $d(u)=d(v)$, then any edge-irregulator of $G$ contains at least one edge incident to $u$ or $v$. Also, any vertex-irregulator of $G$ contains at least one vertex in $N(u)\cup N(v)$.
\end{observation}

\begin{observation}\label{obs:twins}
If $G$ contains two vertices $u,v\in V$ that are twins, \emph{i.e.}, $N(u)\setminus\{v\}=N(v)\setminus \{u\}$, such that $uv\in E$, then any vertex-irregulator of $G$ contains at least one vertex in $\{u,v\}$.
\end{observation}

The importance of upcoming Lemma~\ref{lem:reduction-tw-c} lies in the fact that we can repeatedly apply it, reducing the size of the graph on which we are searching for a vertex-irregulator. 
This is a core argument behind the algorithms of Theorems~\ref{thm:FPT-neighborhood diversity} and~\ref{thm:FPT-cluster}.

\begin{lemma}\label{lem:reduction-tw-c}
Let $G=(V,E)$ be a graph and $u,v\in V$ be a pair of adjacent twins. Let $G'=(V',E')$ be the graph resulting from the deletion of either $u$ or $v$ from $G$. Then, $\I_v(G)=\I_v(G')+1$.
\end{lemma}
\begin{proof}
Assume w.l.o.g. that $u\notin V'$. We first prove that $\I_v(G)\leq \I_v(G')+1$. Indeed, assume that $\I_v(G)>\I_v(G')+1$ and let $S'$ be an optimal vertex-irregulator of $G'$. Next, consider the graph $\Tilde{G}=G[V\setminus (S'\cup \{u\})]$. From the construction of $G'$, it follows that $\Tilde{G}=G'[V'\setminus S']$. Since $S'$ is a vertex-irregulator of $G'$, we obtain that $\Tilde{G}$ is locally irregular. In other words, the set $S'\cup \{u\}$ is a vertex-irregulator of $G$ and $|S'\cup \{u\}|=I_v(G')+1$, a contradiction. 

Next, assume that $\I_v(G)<\I_v(G')+1$ and let $S$ be an optimal vertex-irregulator of $G$. It follows from Observation~\ref{obs:twins} that $|\{u,v\}\cap S|\geq 1$. Assume w.l.o.g. that $u\in S$. By the construction of $G'$, we have that $G'[V'\setminus (S\setminus \{u\})]=G[V\setminus S]$ and the set $S\setminus \{u\}$ is a vertex-irregulator of $G'$. In other words, $\I_v(G')\leq |S|-1=\I_v(G)-1$, a contradiction.
\end{proof}

We close this section with some observations on the proof that computing $\I_v(G)$ is $\W[1]$-hard parameterised by the treewidth of $G$, initially presented in~\cite{FMT22}, which allows us to show that this result holds even if we consider more ``restricted'' parameters, such as the treedepth or the feedback vertex set number (\emph{i.e.}, size of a minimum feedback vertex set) of the input graph. Recall that the \emph{treedepth} of a graph $G=(V, E)$ can be defined recursively: if $|V|=1$, then $G$ has treedepth $1$. Then, $G$ has treedepth $k$ if there exists a vertex $v\in V$ such that every connected component of $G[V\setminus\{v\}]$ has treedepth at most $k-1$. Given a graph $G$ and a tree $T$ rooted at a vertex $u$, by \emph{attaching} $T$ on a vertex $v$ of $G$, we mean the operation of adding $T$ to $G$ and identifying $u$ with $v$, \textit{i.e.}, $V(T)\cap V(G)=\{u\}=\{v\}$.

\begin{observation}\label{obs:bounded-td-fvs}
Let $G$ be a graph with vertex cover number (\emph{i.e.}, size of a minimum vertex cover) $k_1$ and $T$ be a rooted tree of depth $k_2$. Let $G'=(V',E')$ be the graph after attaching an arbitrary number of copies of $T$ directly on vertices of $G$. Then $G'$ has treedepth $\mathcal{O}(k_1+k_2)$ and feedback vertex set number $\mathcal{O}(k_1)$.
\end{observation}

Indeed, if $U$ is a minimum vertex cover of $G$, then $G'[V'\setminus U]$ is a union of trees, each tree of depth at most $k_2$. Thus, $U$ is a feedback vertex set of $G'$. Also, any tree of depth at most $k_2$ is of treedepth at most $k_2$. 

The reduction presented in~\cite[Theorem~$16$]{FMT22} starts with a graph $G$ which is part of an instance of the \textsc{List Colouring} problem, and constructs a graph $G'$ by attaching some trees of depth at most $3$ on each vertex of $G$. The \textsc{List Colouring} problem was shown to be $\W[1]$-hard in~\cite{FLRSST11} when parameterised by the vertex cover number of the input graph. Thus, by Observation~\ref{obs:bounded-td-fvs}, we obtain the following:

\begin{corollary}
Given a graph $G$, it is $\W[1]$-hard to compute $\I_v(G)$ parameterised by either the treedepth or the feedback vertex set number of $G$. 
\end{corollary}

\section{FPT algorithms for vertex-irregulators}\label{sec:FPT-general}


In this section we present two FPT algorithms that compute an optimal vertex-irregulator of a given graph $G$, when parameterised by the neighbourhood diversity or the vertex integrity of $G$. The latter algorithm is then used to show that this problem is in FPT also when parameterised by the cluster deletion number of $G$. We begin by recalling the needed definitions.

The \emph{twin equivalence} of $G$ is the relation on the vertices of $V$ where two vertices belong to the same equivalence class if and only if they are twins. 

\begin{definition}[\cite{La12}]
A graph $G$ has \emph{neighbourhood diversity} $k$ (denoted as $nd(G)=k$) if its twin equivalence has $k$ classes. 
\end{definition}

Let $G=(V,E)$ be a graph with $nd(G)=k$ and let $V_1,\dots,V_k$ be the partition of $V$ defined by the twin equivalence of $G$. Observe that for any $i\in [k]$, we have that  $G[V_i]$ is either an independent set or a clique.

\begin{theorem}\label{thm:FPT-neighborhood diversity}
Given a graph $G=(V,E)$ such that $nd(G)=k$, there exists an algorithm that computes $\I_v(G)$ in FPT-time parameterised by $k$.
\end{theorem}

\begin{proof}
Let $V_1,\dots,V_k$ be the partition of $V$ defined by the twin equivalence of $G$. Note that this partition can be computed in linear time~\cite{La12}. We begin by constructing an induced subgraph $G'=(V',E')$ of $G$ by applying the following procedure: for each $i\in [k]$, if $G[V_i]$ is a clique on at least two vertices, delete all the vertices of $V_i$ except one; let $D$ be the set of vertices that were deleted and $d=|D|$. This procedure terminates after $k$ iterations and, thus, runs in polynomial time. 
Moreover, it follows from Lemma~\ref{lem:reduction-tw-c} that $\I_v(G)=\I_v(G')+d$. Thus, it suffices to solve the problem on $G'$. For every $i\in [k]$, let $V'_i=V_i\cap V'$. 

Observe that for every locally irregular graph $H$, there exists a prime number $p$ such that $d_H(u)-d_H(v)\neq 0\bmod p$ for every $uv\in E(H)$. In our case, since for every $uv\in E'$ we have that $u\in V'_i$ and $v\in V'_j$ for $i<j\leq [k]$, it follows that
there can be at most $k \choose 2$ possible differences modulo $p$ between the degrees of adjacent vertices in $G^*$, where $G^*=(V^*,E^*)$ is a locally irregular induced subgraph of $G'$. 

We claim that $p\leq (k^2\log n +1) (\log (k^2\log n+1)+\log \log (k^2\log n+1)-\frac{1}{2})$. Indeed, since each one of the differences we considered in the previous paragraph is at most $n$, each one of them has at most $\log n$ prime divisors. Thus, $p$ is at most the $k^2\log n+1^{th}$ prime number. The claimed inequality then follows from the classical result from~\cite{RS62}, according to which the $n^{th}$ prime number is at most $n(\log n + \log\log n -\frac{1}{2})$.

For every $i\in [k]$, let $V^*_i=V'_i\cap V^*$. For every such prime $p$, we consider all the possible cases for $p_i=|V^*_i|\bmod p$, for every $i\in[k]$; there are at most $p^k$ such instances. 
Let us consider any such instance such that $d_{G^*}(u)-d_{G^*}(v)\neq 0\bmod p$ for every $uv\in E^*$. Checking this inequality is straightforward from the $p_i$s. We store the maximum orders of the $V^*_i$s such that $|V^*_i|\bmod p=p_i$ for every $i\in[k]$. Having repeated this procedure for all such instances, we are certain to have computed a locally irregular induced subgraph of $G'$ of maximum order. 
In total, this procedure takes time $p^{k+1}n^{\mathcal{O}(1)}$ 
which is in FPT due to the upper bound on $p$ by~\cite{RS62} and since $\log^k n\leq f(k)n$, for some computable function $f$~\cite{JLRSS17}.
\end{proof}

We now present an FPT algorithm to compute an optimal vertex-irregulator of an input graph $G$ when parameterised by the vertex integrity of $G$, which can be computed in FPT-time~\cite{DDH16}. 

\begin{definition}
A graph $G=(V,E)$ has \emph{vertex integrity} $k$ if there exists a set $U \subseteq V$ such that $|U| = k' \le k$ and all connected components of $G[V\setminus U]$ are of order at most $k - k'$.
\end{definition}
 
\begin{theorem}\label{thm:fpt-vi}
Given a graph $G=(V,E)$ with vertex integrity $k$, there exists an algorithm that computes $\I_v(G)$ in FPT-time parameterised by $k$.
\end{theorem}

\begin{proof}
Let $U\subseteq V$ be such that $|U|=k'\leq k$ and $C_1,\ldots, C_m$ be the vertex sets of the connected components of $G[V\setminus U]$ such that $|C_j|\leq k-k'$, $j\in[m]$. Assume that we know the intersection of an optimal vertex-irregulator $S$ of $G$ and the set $U$, and let $S' = S \cap U$ and $U' = U \setminus S$ (there are at most $2^{|U|} \le 2^k$ possible intersections $S'$ of $U$ and $S$). 
Notice that the graph $G[V\setminus S']$ has an optimal vertex-irregulator that contains only vertices from $\bigcup_{i \in [m]}C_i$. Indeed, assuming otherwise contradicts that $S'$ is the intersection of an optimal vertex-irregulator and $U$. Thus, in order to find an optimal vertex-irregulator $S$ of $G$, it suffices to compute $S^* \subseteq \bigcup_{i \in [m]}C_i$, which is an optimal vertex-irregulator of $G[V \setminus S']$, for every set $S' \subseteq U$. Then, we return the set $S^*\cup S'$ of minimum order. We compute $S^*$ through an ILP with bounded number of variables. To do so, we define types and sub-types of graphs $G[U'\cup C_j]$, $j\in[m]$.


Informally, the main idea is to categorise the graphs $G[U' \cup C_j]$, $j \in [m]$, into types based on their structure (formally defined later), whose number is bounded by some function of $k$. Each type $i$ is associated with a number $no_i$ that represents the number of the subgraphs $G[U' \cup C_j]$, $j \in [m]$, that belong in that type. 
Then, for each type $i$, we will define sub-types based on the induced subgraphs $G[(U' \cup C_j) \setminus S_q]$, for $S_q \subseteq C_j$. We also define a variable $no_{i,q}$ that is the number of the subgraphs $G[U' \cup C_j]$, $j \in [m]$, that are of type $i$ and of sub-type $q$ in $G[V\setminus S]$. 
Note that knowing the structure of these types and sub-types, together with $no_{i,q}$, is enough to compute the order of $S^*$. Finally, for any $j \in [m]$, the graph $G[U' \cup C_j]$ is of order at most $k$. Thus, the number of types, sub-types and their corresponding variables, is bounded by a function of $k$. We will present an ILP formulation whose objective is to minimise the order of $S^*$.


We begin by defining the types. Recall first that a function $f:A\rightarrow B$ is a \emph{bijection} if, for every $a_1,a_2\in A$ with $a_1\neq a_2$, we have that $f(a_1)\neq f(a_2)$ and for every $b\in B$, there exists an $a\in A$ such that $f(a)=b$. Recall also that the \emph{inverse} function of $f$, denoted as $f^{-1}$, exists if and only if $f$ is a bijection, and is such that $f^{-1}:B\rightarrow A$ and for each $b\in B$ we have that $f^{-1}(b)=a$, where $f(a)=b$. Two graphs $G[U' \cup C_i]$ and $G[U' \cup C_j]$, $i,j \in [m]$, are of the same type if there exists a bijection 
$f: C_i\cup U' \rightarrow C_j\cup U'$ such that $f(u)=u$ for all $u\in U'$ and $N_{G[U' \cup C_i]}(u) = \{ f^{-1}(v) \mid v \in N_{G[U' \cup C_j]}(f(u))\}$ for all $u \in C_i$. Note that if such a function exists, then $G[U' \cup C_i]$ is isomorphic to $G[U' \cup C_j]$.

Let $p$ be the number of different types. Notice that $p$ is bounded by a function of $k$ as any graph $G[U' \cup C_i]$ has order at most $k$. Also, we can decide if two graphs $G[U' \cup C_i]$ and $G[U' \cup C_j]$, $i,j \in [m]$, are of the same type in FPT-time parameterised by $k$. For each type $i \in [p]$, set $no_i$ to be the number of graphs $G[U' \cup C_j]$, $j \in [m]$, of type $i$. 
Furthermore, for each type $i \in [p]$ we select a $C_j$, $j \in [m]$, such that $G[U' \cup C_j]$ is of type $i$, to represent that type; we will denote this set of vertices by $R_i$. 

We are now ready to define the sub-types. 
Let $i \in [p]$ be a type represented by $R_i$ and $S^i_1,\ldots , S^i_{2^{|R_i|}}$
be an enumeration of the subsets of $R_i$. 
For any $q \in [2^{|R_i|}]$, we define a sub-type $(i,q)$ which represents the induced subgraph $G[(U' \cup R_i) \setminus S_q^i ]$. Let $no_{i,q}$ be the variable corresponding to the number of graphs represented by $G[U'\cup R_i]$, $i\in[p]$, that is of type $(i,q)$ in $G[V\setminus S^*]$, for a vertex-irregulator $S^*$ such that $S^* \cap R_i = S^i_q$. 

Notice that, given a vertex-irregulator $S^* \subseteq \bigcup_{j \in [m]} C_{j}$ of $G[V \setminus S']$, there exists a sub-type $(i,q)$, $i\in [p]$, $q\in [2^{|R_i|}]$, for each  $j\in [m]$, such that the graph $G[(U' \cup C_j)\setminus S^*]$ is of sub-type $(i,q)$. Also, assuming that we know the order of $|S^i_q|$ and the number $no_{i,q}$ for all $i\in [p]$, $q\in [2^{|R_i|}]$, then $|S^*| = \sum_{i \in [p]} \sum_{q \in [2^{|R_i|}]} no_{i,q} |S^i_q|$. 

Before giving the ILP formulation whose goal is to find a vertex-irregulator $S^*$ while minimising the above sum, we guess the $(i,q)$ such that $no_{i,q}\neq 0$. 
Let $S_2$ be the set of pairs $(i,q)$, $i \in [p]$ and $q \in [2^{|R_i|}]$, such that there are two vertices $u,v \in R_i \setminus S^i_q $ where $uv\in E(G[(U'\cup C_i)\setminus S^i_q ])$ and $d_{G[(U'\cup R_i)\setminus S^i_q ]}(u) =d_{G[(U'\cup R_i)\setminus S^i_q ]}(v)$. For every $(i,q)\in S_2$, we have that $no_{i,q}=0$. Indeed, assuming otherwise contradicts the fact that $S^*$ is a vertex-irregulator. 
We guess $S_1 \subseteq \{ (i,q) \mid  i \in [p], q \in 2^{|R_i|} \} \setminus S_2 $ such that $no_{i,q} \neq 0 $ for all $(i,q) \in S_1$. Observe that the number of different sets that are candidates for $S_1$ is bounded by some function of $k$.

\bigskip 

\hrule
\smallskip

\noindent Constants

\begin{tabular}{lll}
 $no_i$ & $ i \in [p]$ & number of components of type $i$ \\
& & \\
$e_{uv}\in \{0,1\}$ & $u,v\in U'$ & set to $1$ iff $uv \in E(G[U'])$ \\
& & \\
$e^{i,q}_{u,v}\in \{0,1\}$ & $i\in [p], q\in [2^{|R_i|}]$, $u \in U' $ & set to $1$ iff $uv \in E(G[(U' \cup R_i) \setminus S^i_{q}])$ \\
& and  $v\in R_i \setminus S^i_{q}$ & \\
& & \\
$b^{i,q}_{u}\in [n]$ & $i\in [p], q\in [2^{|R_i|}]$ and $u \in U'$ & set to $d_{G[(\{u\} \cup R_i) \setminus S^i_{q}]}(u)$ \\
& & \\
$d^{i,q}_{u}\in [n]$ & $i\in [p], q\in [2^{|R_i|}]$ and $u \in R_i \setminus S^i_q$ & set to $d_{G[(U' \cup R_i) \setminus S^i_{q}]}(u)$\\
& & 

\end{tabular}

\noindent Variables

\begin{tabular}{lll}
$no_{i,q}$ & $i \in [p], q\in [2^{|R_i|}]$ & number of graphs of types $(i,q)$ 
\end{tabular}

\bigskip 

\noindent Objective

\begin{align}
\min {\sum}_{i \in [p]} {\sum}_{q \in [2^{|R_i|} ]} no_{i,q} |S^i_{q}|
\end{align}

\noindent Constraints
\bigskip 
\begin{align}
& no_{i,q}=0 && \text{ iff } (i,q)\notin S_1\\
&\sum_{q \in [2^{|R_i|}]} no_{i,q} = no_{i} && \forall i \in [p]\\
&\sum_{w \in U'} e_{wv} + \sum_{i \in [p]} no_{i,q} b^{i,q}_{v} \neq \sum_{w \in U'} e_{wu} + \sum_{i \in [p]} no_{i,q} b^{i,q}_{u}  && \forall u,v \in U', e_{uv}=1\label{model:vi2} \\
&d^{i,q}_{v} \neq \sum_{w \in U'} e_{wu} + \sum_{i \in [p]} no_{i,q} b^{i,q}_{u}  && \forall e^{i,q}_{u,v} = 1 \text{ and } (i,q) \in S_1\label{model:vi3} 
\end{align}

\hrule

\bigskip 

Assume that we have found the values $no_{i,q}$ for $(i,q)$, $i\in [p]$, $q\in [2^{|R_i|}]$. 
We construct an optimal vertex-irregulator of $G[V\setminus S']$ as follows. 
Start with an empty set $S^*$.
For each $i \in [p]$ take all components $C_j$ of type $i$. 
Partition them into $2^{|R_i|}$ many sets $\mathcal{C}^i_q$ each of which containing exactly $no_{i,q}$ of these components, for each $q \in [2^{|R_i|}]$.
For any component in $\mathcal{C}^i_q$, select all vertices represented by the set $S^i_q$ (as it was defined before) and add them to $S^*$.
The final $S^*$ is an optimal vertex-irregulator for $G[V\setminus S']$. 

Let $S=S'\cup S^*$. We show that $S$ is a vertex-irregulator of $G$.
To do so, it suffices to verify that in the graph $G[V\setminus S]$ there are no two adjacent vertices with the same degree. 
Let $u,v$ be a pair of adjacent vertices in a component represented by $R_i \setminus S$, which is of type $(i,q)$. 
If $d_{G[V\setminus S]}(u) = d_{G[V\setminus S]}(v)$, then $(i,q)\in S_2$. Therefore, $no_{i,q}= 0$ and we do not have such a component in $G[V\setminus S]$. 
Thus, it suffices to focus on adjacent vertices such that at least one of them is in $U'$. 
Notice that, in $G[V\setminus S]$, the degree of vertex $u \in U'$ is equal to $\sum_{w \in U'} e_{wv} + \sum_{i \in [p]} no_{i,q} b^{i,q}_{v}$. In other words, no two adjacent vertices in $U'$ have the same degree due to the constraint~\ref{model:vi2}.
Lastly, the constraint~\ref{model:vi3} guarantees that no vertex in $U'$ is adjacent to a vertex in $C_i \setminus S$ (for some $i\in [p]$) such that both of them have the same degree in $G[V\setminus S]$. Moreover, both $S'$ and $S^*$ are constructed to be minimum such sets. Thus, $S$ is an optimal vertex-irregulator of $G$. Finally, since the number of variables in the model is bounded by a function of $k$, we can obtain $S^*$ in FPT time parameterised by $k$ (by running for example the Lenstra algorithm~\cite{Len83}).
\end{proof}

The previous algorithm can be used to find an optimal vertex-irregulator of a graph $G$ in FPT-time when parameterised by the cluster deletion number of $G$. Note that the cluster deletion number of a graph can be computed in FPT-time parameterised by $k$~\cite{HKMN10}.
\begin{definition}[\cite{HKMN10}]
A graph $G=(V,E)$ has \emph{cluster deletion number} $k$ if there exists a set $S\subseteq V$ such that all the connected components of $G[V\setminus S]$ are cliques, and $S$ is of order at most $k$.
\end{definition}

\begin{theorem}\label{thm:FPT-cluster}
Given a graph $G=(V,E)$ with cluster deletion number $k$, there exists an algorithm that computes $\I_v(G)$ in FPT-time parameterised by $k$.
\end{theorem}

\begin{proof}
    Let $S$ be such that $|S|=k$ and $G[V\setminus S]$ is a disjoint union of cliques $C_1,\dots C_m$ for $m\geq 1$. Our goal is to reduce the size of these cliques so that each one of them has order at most $2^k$. We achieve this through the following procedure. Let $i\in[m]$ be such that the clique $C_i=(V_{C_i},E_{C_i})$ has $|V_{C_i}|>2^k$. Let $V_1,\dots,V_p$ be the partition of $V_{C_i}$ defined by the twin equivalence of $G[V_{C_i}\cup S]$. That is, two vertices $u,v\in V_{C_i}$ belong in a $V_j$, $j\in[p]$, if and only if $u$ and $v$ are twins. Note that $p\leq 2^k$. Observe that, since $C_i$ is a clique, the graphs $C_i[V_j]$, $j\in[p]$, are also cliques. In other words, for each $j\in[p]$, all the vertices of $V_j$ are adjacent twins. We delete all but one vertex of $V_j$, for each $j\in[p]$, and repeat this process for every $i\in[m]$ such that $|V_{C_i}|>2^k$.
    Let $G'=(V',E')$ be the resulting subgraph of $G$ and $d=|D|$, where $D$ is the set of vertices that were removed throughout this process. It follows from Lemma~\ref{lem:reduction-tw-c} that $\I_v(G)=\I_v(G')+d$. Observe also that $S\subseteq V'$ and that each connected component of $G'[V'\setminus S]$ is a clique of at most $2^k$ vertices. In other words, $G'$ has vertex integrity at most $2^k+k$. To sum up, to compute $\I_v(G)$ it suffices to compute $\I_v(G')$, which can be done in FPT-time by running the algorithm presented in Theorem~\ref{thm:fpt-vi}.
\end{proof}

\section{Edge-irregulators}\label{sec:hardness}

In this section we begin the study of finding an optimal edge-irregulator of a given graph. It turns out that the decision version of this problem is \np-complete, even for quite restrictive classes of graphs (see Theorem~\ref{thm:hardness}). Furthermore, it is also $\W[1]$-hard parameterised by the size of the solution. 

\begin{theorem}\label{thm:hardness}
Let $G$ be a graph and $k\in \mathbb{N}$. Deciding if $\I_e(G)\leq k$ is \np-complete, even if $G$ is a planar bipartite graph of maximum degree $6$. 
\end{theorem}

\begin{proof}
The problem is clearly in \np. We focus on showing it is also \np-hard. This is achieved through a reduction from the \textsc{Planar 3-SAT} problem which is known to be \np-complete~\cite{L82}. In that problem, a 3CNF formula $\phi$ is given as an input. We say that a bipartite graph $G'=(V,C,E)$ \emph{corresponds} to $\phi$ if it is constructed from $\phi$ in the following way: for each literal $x_i$ (resp. $\lnot x_i$) that appears in $\phi$, add the \emph{literal vertex} $v_i$ (resp. $v'_i$) in $V$ (for $1\leq i\leq n$) and for each clause $C_j$ of $\phi$ add a \emph{clause vertex} $c_j$ in $C$ (for $1\leq j\leq m$). Then the edge $v_ic_j$ (resp. $v'_ic_j$) is added if the literal $x_i$ (resp. $\lnot x_i$) appears in the clause $C_j$. Finally, we add the edge $v_iv'_i$ for every $i$.  A 3CNF formula $\phi$ is valid as input to the \textsc{Planar 3-SAT} problem if the graph $G'$ that corresponds to $\phi$ is planar. Furthermore, we may assume that each variable appears in $\phi$ twice as a positive and once as a negative literal~\cite{dahlhaus1994complexity_cuts}. The question is whether there exists a truth assignment to the variables of $X$ satisfying $\phi$. Starting from a 3CNF formula $\phi$, we construct a graph $G$ such that $\I_e(G)\leq 3n$ if and only if $\phi$ is satisfiable.

\begin{figure}[!t]
\centering
\begin{tikzpicture}[scale=0.6, inner sep=0.5mm]

\begin{scope}[xshift=-4cm]

\node[draw, circle, line width=0.5pt, fill=white](u1) at (0,1.5)[] {};
\node[draw, circle, line width=0.5pt, fill=white](u2) at (0,0.5)[] {};
\node[draw, circle, line width=0.5pt, fill=white](u3) at (1,1.5)[] {};
\node[draw, circle, line width=0.5pt, fill=white](u4) at (1,0.5)[label=above: {\footnotesize $u_1$}] {};
\node[draw, circle, line width=0.5pt, fill=white](u5) at (1,-1)[label=above: {\footnotesize $u_3$}] {};
\node[draw, circle, line width=0.5pt, fill=white](u6) at (2,1)[label=above: {\footnotesize $u_2$}] {};
\node[draw, circle, line width=0.5pt, fill=white](u7) at (2,-1)[label=above: {\footnotesize $u_4$}] {};
\node[draw, circle, line width=0.5pt, fill=white](u8) at (3,0)[label=above: {\footnotesize $u_5$}] {};
\node[draw, circle, line width=0.5pt, fill=white](u9) at (5,-1)[label=above: {\footnotesize $v'_i$}] {};
\node[draw, circle, line width=0.5pt, fill=white](u11) at (5,1)[label=above: {\footnotesize $v_i$}] {};

\node[draw, circle, line width=0.5pt, fill=white](u10) at (2,-2)[] {};
\node[draw, circle, line width=0.5pt, fill=white](u12) at (7,0)[label=above: {\footnotesize $u_6$}] {};
\node[draw, circle, line width=0.5pt, fill=white](u13) at (8,-1)[label=below: {\footnotesize $u_8$}] {};
\node[draw, circle, line width=0.5pt, fill=white](u14) at (8,1)[label=above: {\footnotesize $u_7$}] {};
\node[draw, circle, line width=0.5pt, fill=white](u15) at (9,-1)[label=below:{\footnotesize $u_{12}$}] {};
\node[draw, circle, line width=0.5pt, fill=white](u16) at (9,0)[label=below: {\footnotesize $u_{11}$}] {};
\node[draw, circle, line width=0.5pt, fill=white](u17) at (9,1)[label=above: {\footnotesize $u_{10}$}] {};
\node[draw, circle, line width=0.5pt, fill=white](u18) at (9,2)[label=above: {\footnotesize $u_9$}] {};
\node[draw, circle, line width=0.5pt, fill=white](u19) at (10,-0.25)[] {};
\node[draw, circle, line width=0.5pt, fill=white](u20) at (10,0.25)[] {};
\node[draw, circle, line width=0.5pt, fill=white](u21) at (10,0.75)[] {};
\node[draw, circle, line width=0.5pt, fill=white](u22) at (10,1.25)[] {};
\node[draw, circle, line width=0.5pt, fill=white](u23) at (10,1.75)[] {};
\node[draw, circle, line width=0.5pt, fill=white](u24) at (10,2.25)[] {};

\draw[-, line width=0.5pt]  (u1) -- (u3);
\draw[-, line width=0.5pt]  (u2) -- (u4);
\draw[-, line width=0.5pt]  (u3) -- (u6);
\draw[-, line width=0.5pt]  (u4) -- (u6);
\draw[-, line width=0.5pt]  (u5) -- (u7);
\draw[-, line width=0.5pt]  (u6) -- (u8);
\draw[-, line width=0.5pt]  (u7) -- (u8);
\draw[-, line width=0.5pt]  (u8) -- (u11);
\draw[-, line width=0.5pt]  (u8) -- (u9);

\draw[-, line width=0.5pt]  (u9) --node [fill=white] {\footnotesize $e_i^4$} (u10);
\draw[-, line width=0.5pt]  (u9) -- (u12);
\draw[-, line width=0.5pt]  (u11) -- (u12);
\draw[-, line width=0.5pt]  (u12) -- (u13);
\draw[-, line width=0.5pt]  (u12) -- (u14);
\draw[-, line width=0.5pt]  (u13) -- (u15);
\draw[-, line width=0.5pt]  (u14) -- (u16);
\draw[-, line width=0.5pt]  (u14) -- (u17);
\draw[-, line width=0.5pt]  (u14) -- (u18);
\draw[-, line width=0.5pt]  (u16) -- (u19);
\draw[-, line width=0.5pt]  (u16) -- (u20);
\draw[-, line width=0.5pt]  (u17) -- (u21);
\draw[-, line width=0.5pt]  (u17) -- (u22);
\draw[-, line width=0.5pt]  (u18) -- (u23);
\draw[-, line width=0.5pt]  (u18) -- (u24);

\draw (u11) [dash pattern=on 2pt off 2pt]  -- node [fill=white] {\footnotesize $e_i^2$} (8,2.5);
\draw (u11) [dash pattern=on 2pt off 2pt]  -- node [fill=white] {\footnotesize $e_i^1$} (8,4);
\draw (u9)[dash pattern=on 2pt off 2pt]  -- node [fill=white] {\footnotesize $e_i^3$} (8,-2);

\node at (5, -3) {\small(a): The variable gadget};
\end{scope}

\begin{scope}[xshift=10cm]

\node[draw, circle, line width=0.5pt, fill=white](c) at (0,0)[label=above right: {\footnotesize $c_j$}] {};
\node[draw, circle, line width=0.5pt, fill=white](v1) at (0,1)[] {};
\node[draw, circle, line width=0.5pt, fill=white](v2) at (0,-1)[] {};
\node[draw, circle, line width=0.5pt, fill=white](u1) at (1,0)[] {};
\node[draw, circle, line width=0.5pt, fill=white](v3) at (1,1)[] {};
\node[draw, circle, line width=0.5pt, fill=white](v4) at (1,-1)[] {};
\node[draw, circle, line width=0.5pt, fill=white](u2) at (2,1)[] {};
\node[draw, circle, line width=0.5pt, fill=white](u3) at (2,0)[] {};
\node[draw, circle, line width=0.5pt, fill=white](u4) at (2,-1)[] {};

\draw[-, line width=0.5pt]  (c) -- (u1);
\draw[-, line width=0.5pt]  (u1) -- (u2);
\draw[-, line width=0.5pt]  (u1) -- (u3);
\draw[-, line width=0.5pt]  (u1) -- (u4);
\draw[-, line width=0.5pt]  (c) -- (v1);
\draw[-, line width=0.5pt]  (c) -- (v2);
\draw[-, line width=0.5pt]  (u1) -- (v3);
\draw[-, line width=0.5pt]  (u1) -- (v4);
\draw [dash pattern=on 2pt off 2pt] (-1,1)-- (c);
\draw [dash pattern=on 2pt off 2pt] (-1,0)-- (c);
\draw [dash pattern=on 2pt off 2pt] (-1,-1)-- (c);

\node at (0.5, -2) {\small(b): The clause gadget};

\end{scope}
\end{tikzpicture}
\caption{The construction in the proof of Theorem~\ref{thm:hardness}. The dashed lines are used to represent the edges between the literal and the clause vertices.}\label{fig:np-c}
\end{figure}

\medskip 

\noindent\textbf{Construction.} We start with the graph $G'$ that corresponds to the formula $\phi$. Then, for each $1\leq i\leq n$, we remove the edge $v_iv'_i$, and attach the gadget illustrated in Figure~\ref{fig:np-c}(a) to $v_i$ and $v'_i$ (the variable gadget). Let $E_i$ denote the  edges of the gadget attached to $v_i$ and $v'_i$ plus the edges $e_i^1,e_i^2$ and $e_i^3$. The edges $e_i^1$ and $e_i^2$ connect $v_i$ to the two clause vertices that correspond to clauses containing $x_i$. The edge $e_i^3$ connects $v_i'$ to the one clause vertex that corresponds to the clause containing $\lnot x_i$. Finally, for each $1\leq j\leq m$, we attach two leaves to $c_j$ and then we add the star with $7$ vertices and identify one of its leaves as the vertex $c_j$ (the clause gadget, illustrated in Figure~\ref{fig:np-c}(b). Observe that the resulting graph $G$ is planar, bipartite and $\Delta(G)=6$. 

\medskip
\noindent\textbf{Properties.} Before we provide the reduction, let us show two claims that are going to be useful.

\begin{claimOC}\label{cl:each-gadget-3}
Let $S$ be an edge-irregulator of $G$. For every $1\leq i\leq n$, we have that $|S\cap E_i|\geq 3$.
\end{claimOC}

\begin{proofclaim}
Observe that $d_G(u_5)=d_G(v_i)=d_G(u_6)=d_G(u_7)=d_G(v'_i)$. It follows that $S$ contains at least one edge $a_1$ incident to $u_6$ or $u_7$ and one edge $a_2$ incident to $v'_i$ or $u_5$. We distinguish cases:
\begin{enumerate}
    \item $a_1=a_2=v'_iu_6$. Then $S$ also contains an edge $a_3$ incident to $v_i$ or $u_5$. If $a_3=u_5v'_i$, then $S$ contains an additional edge incident to $u_2$ or $u_5$. If $a_3=u_5u_4$, then $S$ also contains the edge $u_3u_4$. If $a_3=u_2u_5$, then $S$ contains at least one additional edge incident to $u_2$ or $u_1$. If $a_3=u_5v_i$, then $S$ contains one additional edge incident to $u_2$ or $u_5$. In any one of the above cases, we have that $|S\cap E_i|\geq 3$. Thus, we may assume that $a_3$ is incident to $v_i$ but not to $u_5$. If $a_3=e_i^1$ or $a_3=e_i^2$, then $S$ contains an additional edge incident to $v_i$ or $u_6$. Finally, if $a_3=v_iu_6$, then $S$ contains an additional edge incident to $u_6$ or $u_8$. Thus, if $a_1=a_2=v'_iu_6$, then $|S\cap E_i|\geq 3$.
    \item $a_1\neq a_2$. We distinguish some additional cases:
    \begin{enumerate}
        \item $a_1$ is incident to $u_6$ and $a_1=v_iu_6$. 
            \begin{enumerate}
                \item If $a_2\in\{e_i^3,e_i^4\}$, then $S$ contains an additional edge incident to $v'_i$ or $u_6$.  
                \item If $a_2= v'_iu_6$, then $S$ contains an additional edge incident to $u_6$ or $u_8$.
                \item If $a_2$ is incident to $u_5$, then $S$ contains an additional edge incident to $u_2$ or $u_5$.
            \end{enumerate}
        \item $a_1$ is incident to $u_6$ and $a_1=u_6u_7$. Then $S$ contains at least an additional edge incident to, w.l.o.g., $u_9$.
        \item $a_i$ is incident to $u_6$ and $a_1=u_6u_8$. Then $S$ also contains the edge $u_8u_{12}$.
        \item $a_i$ is incident to $u_7$. Then $S$ contains an additional edge incident to $u_7$.
    \end{enumerate}
    Thus, if $a_1\neq a_2$, then $|S\cap E_i|\geq 3$, which finishes the proof of the claim.
\end{enumerate}
\end{proofclaim}

\begin{claimOC}\label{cl:no-bad-deletes}
Let $S$ be an edge-irregulator of $G$ such that $|S|\leq 3n$. Then, for every $1\leq i\leq n$, we have that
\begin{itemize}
    \item if $|S\cap \{e_i^1,e_i^2\}|\geq 1$ then $|S\cap \{e_i^3,e_i^4\}|=0$ and
    \item if $|S\cap \{e_i^3,e_i^4\}|\geq 1$ then $|S\cap \{e_i^1,e_i^2\}|=0$.
\end{itemize}
\end{claimOC}

\begin{proofclaim}
Since the proofs of the two items are highly symmetrical, we will 
only prove the first item. To do that, it suffices to show that if $S$ does not respect the statement for some $1\leq i\leq n$, then $|S\cap E_i|\geq 4$. Then, since $|S|\leq 3n$, and $1\leq i\leq n$, there exists a $1\leq j\leq n$ such that $i\neq j$ and $|S\cap E_j|\leq 2$. This contradicts Claim~\ref{cl:each-gadget-3}.

Let $H=G-S$. Assume first that there exists an $i$ such that, say, $e_i^1\in S$ and $e_i^3\in S$.
Observe that $S$ contains at least one edge $e$ incident to $u_6$ or $u_7$, as otherwise we would have that $d_H(u_6)=d_H(u_7)$, contradicting the fact that $S$ is an edge-irregulator of $G$. Thus, if we also have that $e_i^2\in S$ or that $e_i^4\in S$, it follows that $|S\cap E_i|\geq 4$, a contradiction. Thus, we may assume that $S\cap E_i=\{e_i^1,e_i^3,e\}$. If $e\in\{u_7u_9,u_7u_{10},u_7u_{11}\}$, say $e=u_7u_9$, then $d_H(u_7)=d_H(u_{10})$. Also, if $e=u_6u_8$, then $S$ also contains $u_8u_{12}$. Finally, if $e=v_iu_6$ (resp. $e=u_6v'_i)$ then $d_H(u_6)=d_H(v'_i)$ (resp. $d_H(u_6)=d_H(v_i)$). It follows from Observation~\ref{obs:neigbourhood} that in all cases, we have that $|S\cap E_i|\geq 4$, a contradiction.
\end{proofclaim} \\

We are now ready to give the reduction. 

\medskip 

\noindent\textbf{Reduction.} Let $G$ be the graph constructed from the formula $\phi$ as explained above. We show that there exists a satisfying truth assignment of $\phi$ if and only if $\I_e(G)\leq 3n$. 

For the first direction, let $T$ be a satisfying truth assignment of $\phi$. Let $S$ be the set containing the edges $e_i^1,e_i^2,v_i'u_6$ for every $1\leq i\leq n$ such that $T(x_i)=true$ and the edges $e_i^3,e_i^4,v_iu_6$ for each $i$ such that $T(\lnot x_i)=true$. Let $H=G-S$. Clearly, $|S|=3n$. Also, $S$ is an edge-irregulator of $G$. Indeed, the part of the graph $H$ that corresponds to the gadget attached to $v_i$ and $v'_i$ is clearly locally irregular for every $i$. Also, for each $j$, we have that $d_H(c_j)\leq 5$, since $C_j$ is satisfied by at least one literal. Therefore, the part of the graph $H$ that corresponds to the gadget attached to $c_j$ is locally irregular for every $j$. Finally, observe that in $H$, for any $j$, if $c_j$ is adjacent to $v_i$ (resp. $v_i'$), for any $i$, then $d_H(c_j)\geq 4$ while $d_H(v_i)=3$ (resp. $d_H(v'_i)=3)$) in this case. Thus $H$ is locally irregular.

For the reverse direction, assume that $\I_e(G)\leq 3n$ and let $S$ be an edge-irregulator of $G$ such that $|S|=3n$. Recall that due to Claim~\ref{cl:no-bad-deletes}, for each $i\in[n]$, if $S$ contains one edge in $\{e_i^1,e_i^2\}$ then it contains no edge in $\{e_i^3,e_i^4\}$ and \emph{vice versa}. For each $i\in [n]$, we set $T(x_i)=true$ if $S$ contains one edge in $\{e_i^1,e_i^2\}$ and $T( x_i)=false$ otherwise. 
Here notice that for some variables no edge from the sets $\{e_i^1,e_i^2\}$ or $\{e_i^3,e_i^4\}$ has been included in $S$. The assignment $T$ has set these variables to $false$ as they are included in the ``otherwise''.
Furthermore, these variables do not play an important role in the satisfiability of the formula.
We claim that $T$ is indeed a truth assignment that satisfies $\phi$.

To see this, observe first that due to Claim~\ref{cl:no-bad-deletes}, each variable will receive exactly one truth value, i.e., $T$ is indeed a valid truth assignment over $\{x_i\mid i \in [n]\}$. It remains to show that each clause $C_j$, $j \in [m]$, that appear in $\phi$ is satisfied by $T$. Fix a clause $C_j$ and consider the clause gadget that corresponds to $C_j$.
Due to Claim~\ref{cl:each-gadget-3} and the assumption that $\I_e(G)\leq 3n$, we can conclude that none of the edges of this clause gadget belong to $S$ (otherwise $|S|>3n$). 
Additionally, the vertex $c_j$ has a neighbour of the same degree as itself in $G$ (its neighbour of degree $6$ in the clause gadget). Therefore, there exists an edge incident to $c_j$ that belong to $S$.
Also, by the above arguments, this edge is either $v_ic_j$ or $v'_ic_j$ for some $i \in [n]$. Therefore, there exists an $i\in [n]$ such that either $v_ic_j\in S$ or $v'_ic_j\in S$. 
We claim that the literal of the variable $x_i$ that appears in $C_j$ is a true literal under the assignment $T$.
Assume that $v_ic_j\in S$ (resp. $v'_ic_j\in S$). Since $v_ic_j \in \{e^1_i,e^2_i\}\cap S$ (resp. $v'_ic_j \in \{e^3_i\} \cap S$), due to Claim~\ref{cl:no-bad-deletes} and the definition of $T$, we have that $T(x_i)=true$ (resp. $T(x_i)=false$). 
Additionally, we note that $v_ic_j \in E(G)$ (resp. $v'_ic_j \in E(G)$) which means that the literal $x_i$ appears in $C_j$ (resp.$\neg x_i$ appears in $C_j$). Therefore $C_j$ is satisfied by the literal $x_i$ ($\neg x_i$ reps.).
In other words, each clause of $\phi$ is satisfied by $T$. This ends the reduction. 
\end{proof}

\begin{theorem}\label{thn:edges-w-hard-size-solution}
Let $G$ be a graph and $k\in \mathbb{N}$. Deciding if $\I_e(G)\leq k$ is $\W[1]$-hard parameterised by $k$.
\end{theorem}
\begin{proof}
The reduction is from \textsc{$k$-Multicoloured Clique}. 

\defproblem{\textsc{$k$-Multicoloured Clique}}{
    A graph $G'=(V,E)$ and a partition $(V_1,\ldots,V_k)$ of $V$ into $k$ independent sets. 
}{
    Does there exist a set $S\subseteq V$ such that $|S|=k$ and $G'[S]$ is a clique?
}

It is known that \textsc{$k$-Multicoloured Clique} is $\W[1]$-hard parameterised by $k$~\cite{FLRSST11}, even when $|V_i|=n$ for all $i \in [k]$. 

On a high level, our reduction will proceed as follows. Starting with the graph $G'$ that is given in the input of \textsc{$k$-Multicoloured Clique} with $|V_i|=n$ for all $i \in [k]$, we will subdivide every edge of the graph $G'$. Then, for each $i\in[k]$, we will attach one copy of a particular gadget to the vertices of $V_i$. Also, for each $1\leq i<j\leq k$, we will attach a copy of our gadget to the vertices that correspond to the edges $v_iv_j$ of $G'$, with $v_i\in V_i$ and $v_j\in V_j$. In total, we will add $(k^2+k)/2$ gadgets. 

The gadgets are structured so that any edge-irregulator of the graph contains at least one edge for each gadget (so any solution has a size of at least $(k^2+k)/2$). Furthermore, we prove that, if we have selected only one edge from a gadget, then that edge must be incident to either a vertex of the original graph or a vertex that represents an edge of the original graph.
Finally, we show that:
\begin{itemize}
    \item an edge-irregulator $S$ that contains exactly one edge from each gadget (\textit{i.e.}, an edge-irregulator of size $(k^2+k)/2$) can give us a clique of size $k$ in the original graph by selecting the vertices and edges (represented by vertices) of the original graph that are incident to the edges of $S$ and 
    \item if we have a clique of size $k$ in the original graph we can construct an optimal edge-irregulator $S$ by selecting the edges of the gadgets that are incident to the $k$ vertices of the clique and the $(k^2-k)/2$ vertices that represent the edges of the clique.
\end{itemize}

We proceed with the formal proof. 

\medskip
\noindent\textbf{Construction.} Assume that we are given an instance $G'=(V,E)$ with vertex partition $(V_1,\ldots,V_k)$ where $|V_i| = n$ for all $i \in [k]$. For each $i \in [k]$, we denote by $v_i^p$, for $p\in [n]$, the vertices of $V_i$.
We construct a graph $G$ as follows:
\begin{itemize}
    \item Start with a copy of $G'$.
    \item Subdivide each edge $e \in E$. Let $u_{i,j}^{p,q}$ be the vertex that corresponds to the edge $v_{i}^{p}v_{j}^{q} \in E$. Also, let $U_{i,j}$ be the set of vertices created by subdividing the edges between the sets $V_i$ and $V_j$, \emph{i.e.}, the set $\{u_{i,j}^{p,q} \mid v_{i}^{p}v_{j}^{q} \in E\}$.
    \item For each pair $(i,j)$ where $1\le i < j \le k$, create a copy of the gadget $H_{|U_{i,j}|}$, illustrated in Figure~\ref{figure:size-of-the-solution-reduction}. We denote the copy of $H_{|U_{i,j}|}$ by $H^{w_{i,j}}$, where $w_{i,j}$ is the copy of $w$ that belongs in the gadget, and the copy of $y$ in $H^{w_{i,j}}$ by $y_{i,j}$. Then, we add all the edges between $w_{i,j}$ and the vertices of $U_{i,j}$. 
    \item For each $i \in [k]$, create a copy of the gadget $H_{|V_{i}|}$.  We denote the copy of $H_{|V_{i}|}$ by $H^{w_{i}}$, where $w_{i}$ is the copy of $w$ that belongs in the gadget, and the copy of $y$ in $H^{w_{i}}$ by $y_{i}$. Then, we add all the edges between $w_{i}$ and the vertices of $V_{i}$.
    \item Finally, add leaves attached to the vertices of $V_i$, $i \in [k]$, so that each vertex of $V_i$ has degree $kn$ and attached to the vertices of $U_{i,j}$, $1\le i<j \le k$, so that each vertex of $U_{i,j}$ has degree $kn + 1$. 
\end{itemize} 

Let $G$ be the resulting graph. We proceed with the reduction.

\begin{figure}[!t]
\centering
\scalebox{0.75}{
\begin{tikzpicture}[inner sep=0.7mm]

	\node[draw, circle, line width=1pt, fill=black](v) at (6,2.5)  [label=left: \large$w$]{};

	\node[draw, circle, line width=1pt, fill=black](u) at (8,2.5)  [label=below: \large$y$]{};
        \draw[-, line width=1pt]  (v) -- (u);

\begin{scope}[xshift=0cm,yshift=0cm]   

\draw (6,0) ellipse (2.5cm and 0.5cm);

    \node[draw, circle, line width=1pt, fill=white](bl1) at (4.25,0) []{};
    \node[draw, circle, line width=1pt, fill=white](blx) at (7.75,0) []{};
    \path (bl1) -- (blx) node [black, font=\Large, midway, sloped] {$\,\dots$};

        \draw[-, line width=1pt]  (v) -- (bl1);
        \draw[-, line width=1pt]  (v) -- (blx);
    
    \node[] () at (6,-0.94) []{\large Set $V_{i}$ for $i \in [k]$ or $U_{i,j}$ for $1\le i < j \le k$};


\end{scope}

\begin{scope}[xshift=0cm,yshift=0cm]  

    \node[draw, circle, line width=1pt, fill=black](br1) at (11,0) []{};
    \node[draw, circle, line width=1pt, fill=black](brd) at (13,0) []{};
    \path (br1) -- (brd) node [black, font=\Large, midway, sloped] {$\,\dots$};

    \draw [decorate,decoration={brace, amplitude=10pt},xshift=0pt,yshift=-6pt] (13.25,0) -- (10.75,0) node [black,midway,yshift=-0.71cm] {\large $n^2-2$ leaves};

    \draw[-, line width=1pt]  (u) -- (br1);
    \draw[-, line width=1pt]  (u) -- (brd);

\end{scope}


\begin{scope}[xshift=0cm,yshift=0.5cm]

    \node[draw, circle, line width=1pt, fill=black](rd1) at (10,4.5) []{};
    \node[draw, circle, line width=1pt, fill=black](rd2) at (10,2) []{};

    \draw[-, line width=1pt]  (u) -- (rd1);
    \draw[-, line width=1pt]  (u) -- (rd2);

    \node[draw, circle, line width=1pt, fill=black](rd21) at (12,5.25) []{};
    \node[draw, circle, line width=1pt, fill=black](rd2d) at (12,3.75) []{};
    \path (rd21) -- (rd2d) node [black, font=\Large, midway, sloped] {$\,\dots$};

    \draw[-, line width=1pt]  (rd21) -- (rd1);
    \draw[-, line width=1pt]  (rd2d) -- (rd1);

    \draw [decorate,decoration={brace, amplitude=10pt},xshift=6pt,yshift=0pt] (12,5.5) -- (12,3.5) node [black,midway,xshift=1.7cm] {\large $n^2-1$ leaves};

    \node[draw, circle, line width=1pt, fill=black](rd31) at (12,2.75) []{};
    \node[draw, circle, line width=1pt, fill=black](rd3d) at (12,1.25) []{};
    \path (rd31) -- (rd3d) node [black, font=\Large, midway, sloped] {$\,\dots$};

    \draw[-, line width=1pt]  (rd31) -- (rd2);
    \draw[-, line width=1pt]  (rd3d) -- (rd2);

    \draw [decorate,decoration={brace, amplitude=10pt},xshift=6pt,yshift=0pt] (12,3) -- (12,1) node [black,midway,xshift=1.7cm] {\large $n^2-1$ leaves};

\end{scope}

\begin{scope}[xshift=0.25cm,yshift=-0.25cm] 

    \node[draw, circle, line width=1pt, fill=black](u16) at (4.5,4) []{};
    \node[draw, circle, line width=1pt, fill=black](u17) at (5.5,4) []{};
    \node[draw, circle, line width=1pt, fill=black](u18) at (7,4) []{};
    \node[draw, circle, line width=1pt, fill=black](u19) at (4.5,5) []{};
    \node[draw, circle, line width=1pt, fill=black](u20) at (5.5,5) []{};
    \node[draw, circle, line width=1pt, fill=black](u21) at (7,5) []{};
    
    \draw[-, line width=1pt]  (v) -- (u16);
    \draw[-, line width=1pt]  (v) -- (u17);
    \draw[-, line width=1pt]  (v) -- (u18);
    \draw[-, line width=1pt]  (u16) -- (u19);
    \draw[-, line width=1pt]  (u17) -- (u20);
    \draw[-, line width=1pt]  (u18) -- (u21);
    
    \draw [decorate,decoration={brace, amplitude=10pt},xshift=0pt,yshift=6pt] (4.25,5) -- (7.25,5) node [black,midway,yshift=0.7cm] {\large $n^2-b$ times};
    
    \path (u17) -- (u18) node [black, font=\Large, midway, sloped] {$\,\dots$};
    
\end{scope}

\end{tikzpicture}
}
\caption{The gadget $H_b$, $b \in \mathbb{N}$, used in the proof of Theorem~\ref{thn:edges-w-hard-size-solution}. The black vertices represent the vertices of the gadget. The white vertices represent either a set of the original vertices $V_i$, $i \in [k]$, or a set of edge vertices $U_{i,j}$, $1 \le i < j \le k$. In the construction, if $w$ is adjacent to vertices of a $V_i$, $i \in [k]$, 
then $b=|V_i|$ while if $w$ is adjacent to vertices of a $U_{i,j}$, $1 \le i < j \le k$, then 
$b=|U_{i,j}|$. In each copy of the gadget, the degrees of $w$ and $y$ are equal.}\label{figure:size-of-the-solution-reduction}

\end{figure}

\medskip

\noindent\textbf{Reduction.} We prove that $G$ has an edge-irregulator of order $(k^2 + k)/2$ if and only if $G'$ is a yes instance of \textsc{$k$-Multicoloured Clique}.

Assume that $G'$ is a yes instance of \textsc{$k$-Multicoloured Clique} and $C= \{c_1, \ldots , c_k\}$ is a clique in $G'$ with $c_i\in V_i$ for every $i\in[k]$. We will construct an edge-irregulator of $G$ as follows. Start with an empty set $S$.
Notice that, for each $i \in [k]$, $|V_i \cap C|=1$ and let $p\in[n]$ be such that $v_i^p=c_i$; 
we add to $S$ the edge $v_i^p w_i$. 
For each pair $(i,j)$, $1\le i<j \le k$, let $p,q\in[n]$ be such that $v_i^p=c_i$ and $v_j^q=c_j$; we add to $S$ the edge $u_{i,j}^{p,q}w_{i,j}$. Notice that the edge $v_i^p v_j^q$ must exist in $E$ since $C$ is a clique. It follows that the vertex $u_{i,j}^{p,q}$, and therefore the edge $u_{i,j}^{p,q}w_{i,j}$, also exists in $G$. By construction, $|S| = (k^2+k)/2$. It only remains to prove that $S$ is an edge-irregulator of $G$.

Consider the graph $G-S$. Observe that, for every $H^{w_i}$, $i \in [k]$, we have reduced the degree of $w_i$ by exactly one. Therefore, any two adjacent vertices of $H^{w_i}$ have different degree (see Figure~\ref{figure:size-of-the-solution-reduction}). The same holds true for every $H^{w_{i,j}}$, $1\le i<j \le k$.
Consider now the edges $xz\in E(G)$ such that $x\in\{w_i,w_j,w_{i,j}\}$, and $z\in V_i \cup U_{i,j}\cup V_j$, $1\le i<j \le k$. Notice that $d_{G-S}(x)=n^2-1$ and $kn-1\leq d_{G-S}(z)\leq kn+1$. For sufficiently large $n$, we have that $n^2-1 > kn+1$.
It remains to consider the edges between vertices in $V_i \cup V_j$ and in $U_{i,j}$ for any $1\le i<j \le k$. 
Notice that, for every $1\le i<j \le k$, all vertices of $V_i \cup V_j$, except one vertex $v_i^p\in V_i$ and one vertex $v_j^q\in V_j$, have degree $kn$, and $d_{G-S}(v_i^p)=d_{G-S}(v_j^q)=kn - 1$.
Also, all vertices of $U_{i,j}$, except one vertex $u'$, have degree $kn +1$, and $d_{G-S}(u')=kn$. So, $u'$ is the only vertex of $U_{i,j}$ that could possibly have the same degree as a vertex in $V_i\setminus \{v_i^p\}$ or $V_j\setminus \{v_j^q\}$. It follows by the construction of $S$ that $u'$ is actually $u_{i,j}^{p,q}$. Also, by the construction of $G$,  $u_{i,j}^{p,q}$ is adjacent only to $v_i^p$ and $v_j^q$, as it represents the edge between their corresponding vertices in $G'$. Thus, for every $1\le i<j \le k$, no vertex in $U_{i,j}$ has the same degree as any of its neighbours in $V_i$ or $V_j$. It follows from all the arguments above that $S$ is indeed an edge-irregulator of $G$.

Now we show that if $\I_e(G)=(k^2+k)/2$ then $G'$ has a clique of size $k$. Let $S$ be an edge-irregulator of $G$ of order $(k^2+k)/2$.
First, we notice that for each $i\in [k]$, $d_G(w_i)=d_G(y_i)$ and that for each $1\le i <j \le k$, $d_G(w_{i,j})=d_G(y_{i,j})$. 
Let $E_{w_i}$ be the set of edges $w_iv$ for $v \in V_i$ and $E_{w_{i,j}}$ be the set of edges $w_{i,j}u$ for $u \in U_{i,j}$. Also, let $w \in \{w_i \mid i \in [k]\} \cup \{w_{i,j} \mid 1\le i < j \le k \}$. 
Since $S$ is an edge-irregulator of $G$, it follows that $|S\cap (E(H^w)\cup E_w)|\geq 1$. Also, observe that for any pair of distinct vertices $w, w' \in \{w_i \mid i \in [k]\} \cup \{w_{i,j} \mid 1\le i < j \le k \}$, we have that $(E(H^w)\cup E_w) \cap ( E(H^{w'})\cup E_{w'} ) = \emptyset$. Since $|S|=(k^2+k)/2$, we obtain that, actually, $|S\cap (E(H^w)\cup E_w)|= 1$. Next, we show that $S$ includes only edges from the set $E_w$, for each $w \in \{w_i \mid i \in [k]\} \cup \{w_{i,j} \mid 1\le i < j \le k \}$. In particular we claim the following:

\begin{claimOC}\label{cl:exactly-one}
    Let $w \in \{w_i \mid i \in [k]\} \cup \{w_{i,j} \mid 1\le i < j \le k \}$. 
    It holds that  $S \cap E(H^w) = \emptyset$ and that $|S \cap E_{w}| = 1$. 
\end{claimOC}

\begin{proofclaim}
    Assume that $S\cap E(H^w)\neq \emptyset$ and let $e \in S \cap E(H^w)$. 
    We distinguish cases according to which edge of $H^w$ is $e$. In each case, we show that $S$ must include an additional edge of $E(H^w)$, which is a contradiction to the fact that $|S \cap ( E(H^w) \cup E_{w} ) |= 1$.
    
    \smallskip
    \noindent \textbf{$\boldsymbol{e}$ is incident to neither $\boldsymbol{w}$ nor $\boldsymbol{y}$:} Then $S$ must also include an additional edge incident to $w$ or $y$ (from previous discussion).

    \smallskip
    \noindent\textbf{$\boldsymbol{e}$ is incident to $\boldsymbol{y}$:} Then, $S$ must include an additional edge of $E(H^w)$, as otherwise $d_{G-S}(y)=n^2$ and $y$ would have at least one neighbour of degree $n^2$.

    \smallskip
    \noindent
    \textbf{$\boldsymbol{e}$ is incident to $\boldsymbol{w}$ and $\boldsymbol{e \neq wy}$:} Then, $S$ must include an additional edge of $E(H)$, as otherwise $G-S$ would include a connected component isomorphic to $K_2$. 
\end{proofclaim}\\
The previous claim also shows that $S \subseteq \bigcup_{i \in [k]}E_{w_i} \cup \bigcup_{1\le i < j \le k} E_{w_{i,j}}$.

We now explain how to construct a clique of $G'$ of order $k$.
Let $\ell(i)=m(i)$ be the index that specifies which edge incident to $w_{i}$ is included in $S$. That is, $\ell(i)$ is such that $w_{i} v_i^{\ell(i)} \in S$. 
Similarly, for each $1\le i<j \le k$, let $\ell(i,j)$ and $m(i,j)$ 
be the indices such that $w_{i,j} u_{i,j}^{\ell(i,j) , m(i,j)} \in S$. 
Notice that both $\ell(i)$ and $\ell(i,j)$ are unique as $S$ contains exactly one edge incident to each of $w_i$ and $w_{i,j}$ (by Claim~\ref{cl:exactly-one}).
\begin{claimOC}
    The set $C =\{ v_{i}^{\ell(i)} \mid i \in[k] \}$ induces a clique of order $k$ in $G$.
\end{claimOC}

\begin{proofclaim}
    First, for any $1\le i<j \le k$, we show that $\ell(i) = \ell(i,j)$ and $m(j) = m(i,j)$. To simplify the notation let $\ell = \ell(i,j)$ and $m = m(i,j)$. By the definition of $\ell$ and $m$ we have that $w_{i,j} u_{i,j}^{\ell , m} \in S$.
    Now, we consider the degrees of the vertices $v_{i}^{\ell}$ 
    and $u_{i,j}^{\ell,m}$. 
    Since $w_{i,j} u_{i,j}^{\ell, m } \in S$, we have that $d_{G-S}(u_{i,j}^{\ell,m})=kn$. 
    If $\ell(i) \neq \ell$, 
    then $d_{G-S}(v_{i}^{\ell})=kn $, as $S$ would not include any edges incident to $v_i^{\ell}$ in that case. This is a contradiction since $v_{i}^{\ell} $ and $u_{i,j}^{\ell, m}$ are adjacent in $G$ (by construction) and remain so in $G-S$ (as $S \subseteq \bigcup_{i \in [k]}E_{w_i} \cup \bigcup_{1\le i < j \le k} E_{w_{i,j}}$). Therefore, for any $1\le i<j \le k$, $\ell(i) = \ell = \ell(i,j)$. Similarly, we can show that for any $1\le i<j \le k$, $m(j) = m = m(i,j)$.
    Now we show that for every pair of distinct vertices $u,v \in \{ v_{i}^{\ell(i)} \mid i \in[k] \}$, we have that $u$ and $v$ are adjacent in $G'$.
    W.l.o.g. let $ u =v_{i}^{\ell(i)} $ and $v = v_{j}^{\ell(j)} $ for some $1\le i < j \le k$. We know that $\ell(i) = \ell $ and $\ell (j)  = m(j) = m$. Therefore, the vertex $u_{i,j}^{\ell(i,j) , m(i,j)} = u_{i,j}^{\ell , m}$ of $G$ is adjacent to $v_{i}^{\ell (i)}$ and $v_{j}^{\ell (j)}$. This means that $v_{i}^{\ell (i)}$ and $v_{j}^{\ell (j)}$ are incident in $G'$ as the vertex $u_{i,j}^{\ell(i) , m(j)}$ corresponds to the edge between these two vertices in $G'$ (recall the construction of $G$). 
    Thus, any pair of vertices in $C$ is a pair of adjacent vertices in $G'$. It follows that $C$ is a clique.
\end{proofclaim}\\

This completes the proof.
\end{proof}

Additionally, this problem exhibits a similar behaviour to finding optimal vertex-irregulators, as it also remains intractable even for ``relatively large'' structural parameters.

\begin{theorem}\label{thm:edge-irr-w-hard}
Let $G$ and $k\in \mathbb{N}$. Deciding if $\I_e(G)\leq k$ is $\W[1]$-hard parameterised by either the feedback vertex set number or the treedepth of $G$. 
\end{theorem}

\begin{proof}
The reduction is from the \textsc{General Factor} problem:

\defproblem{\textsc{General Factor}}{
    A graph $H=(V,E)$ and a list function $L: V \rightarrow 2^{\Delta(H)}$ that specifies the available degrees for each vertex $u \in V$. 
}{
    Does there exist a set $S\subseteq E$ such that $d_{H-S}(u) \in L(u)$ for all $u \in V$?
}

This problem is known to be $\W[1]$-hard when parameterised by the vertex cover number of $H$~\cite{GKSSY12}. 

Starting from an instance $(H,L)$ of \textsc{General Factor}, we construct a graph $G$ such that $\I_e(G)< n^2$, where $n=|V(H)|$, if and only if $(H,L)$ is a yes-instance. Moreover, the constructed graph $G$ will have treedepth and feedback vertex set $\mathcal{O}(vc)$, where $vc$ is the vertex cover number of $H$.
For every vertex $u\in V(H)$, let us denote by $\overline{L}(u)$ the set $\{0,1,\ldots, d_H(u)\} \setminus L(u)$. In the case where $\{0,1,\ldots, d_H(u)\} \setminus L(u)=\emptyset$, we set $\overline{L}(u)=\{-1\}$. On a high level, the graph $G$ is constructed by adding some trees on the vertices of $H$. In particular, for each vertex $u\in V(H)$ and for each element $a$ in $\overline{L}(u)$, we will attach a tree to $u$ whose purpose is to prevent $u$ from having degree $a$ in $G-S$, for any optimal edge-irregulator $S$ of $G$. 
We proceed with the formal proof. 

\medskip
\noindent\textbf{Construction.} We begin by defining an arbitrary order on the vertices of $H$. That is, $V(H)=\{u_1,u_2,\dots$ $,u_n\}$. Next, we describe the trees we will use in the construction of $G$. In particular, we will describe the trees that we attach to the vertex $u_i$, for every $1\leq i\leq n$. First, for each $a_j\in \overline{L}(u_i)$, define the value $a'_j=d_H(u_i)-a_j$. Also, for each $j$, let $d_{i,j}=2in^4-a'_j$. 
For each ``forbidden degree'' $a_j$ in the list $\overline{L}(u_i)$, we will attach a tree $T_{i,j}$ to $u_i$. We define the tree $T_{i,j}$ as follows.

\begin{figure}[!t]
\centering
\scalebox{0.7}{
\begin{tikzpicture}[inner sep=0.7mm]

	\node[draw, circle, line width=1pt, fill=white](v1) at (0,0)  []{};
    \node[draw, circle, line width=1pt, fill=white](v2) at (0.5,0)  []{};
    \node[draw, circle, line width=1pt, fill=white](v3) at (1.5,0)  []{};
    \node[draw, circle, line width=1pt, fill=white](v4) at (2.5,0)  []{};
    \node[draw, circle, line width=1pt, fill=white](v5) at (3,0)  []{};
    \node[draw, circle, line width=1pt, fill=white](v6) at (4,0)  []{};
    \node[draw, circle, line width=1pt, fill=white](v7) at (5,0)  []{};
    \node[draw, circle, line width=1pt, fill=white](v8) at (5.5,0)  []{};
    \node[draw, circle, line width=1pt, fill=white](v9) at (6.5,0)  []{};
    \node[draw, circle, line width=1pt, fill=white](v10) at (7.5,0)  []{};
    \node[draw, circle, line width=1pt, fill=white](v11) at (8,0)  []{};
    \node[draw, circle, line width=1pt, fill=white](v12) at (9,0)  []{};
    \node[draw, circle, line width=1pt, fill=white](v13) at (12,0)  []{};
    \node[draw, circle, line width=1pt, fill=white](v14) at (12.5,0)  []{};
    \node[draw, circle, line width=1pt, fill=white](v15) at (13.5,0)  []{};
    \node[draw, circle, line width=1pt, fill=white](v16) at (14.5,0)  []{};
    \node[draw, circle, line width=1pt, fill=white](v17) at (15,0)  []{};
    \node[draw, circle, line width=1pt, fill=white](v18) at (16,0)  []{};
    \node[draw, circle, line width=1pt, fill=white](v19) at (1,1)  []{};
    \node[draw, circle, line width=1pt, fill=white](v20) at (3.5,1)  []{};
    \node[draw, circle, line width=1pt, fill=white](v21) at (5.75,1)  []{};
    \node[draw, circle, line width=1pt, fill=white](v22) at (8.25,1)  []{};
    \node[draw, circle, line width=1pt, fill=white](v23) at (12.75,1)  []{};
    \node[draw, circle, line width=1pt, fill=white](v24) at (15.25,1)  []{};
    \node[draw, circle, line width=1pt, fill=white](v25) at (8,3)  [label=above left: \large$u_{i,j}$]{};
    \node[draw, circle, line width=1pt, fill=white](v26) at (9.5,4)  [label=right: \large$u_i$]{};

    \draw[-, line width=1pt]  (v1) -- (v19);
    \draw[-, line width=1pt]  (v2) -- (v19);
    \draw[-, line width=1pt]  (v3) -- (v19);
    \draw[-, line width=1pt]  (v4) -- (v20);
    \draw[-, line width=1pt]  (v5) -- (v20);
    \draw[-, line width=1pt]  (v6) -- (v20);
    \draw[-, line width=1pt]  (v7) -- (v21);
    \draw[-, line width=1pt]  (v8) -- (v21);
    \draw[-, line width=1pt]  (v9) -- (v21);
    \draw[-, line width=1pt]  (v10) -- (v22);
    \draw[-, line width=1pt]  (v11) -- (v22);
    \draw[-, line width=1pt]  (v12) -- (v22);
    \draw[-, line width=1pt]  (v13) -- (v23);
    \draw[-, line width=1pt]  (v14) -- (v23);
    \draw[-, line width=1pt]  (v15) -- (v23);
    \draw[-, line width=1pt]  (v16) -- (v24);
    \draw[-, line width=1pt]  (v17) -- (v24);
    \draw[-, line width=1pt]  (v18) -- (v24);
    \draw[-, line width=1pt]  (v19) -- (v25);
    \draw[-, line width=1pt]  (v20) -- (v25);
    \draw[-, line width=1pt]  (v21) -- (v25);
    \draw[-, line width=1pt]  (v22) -- (v25);
    \draw[-, line width=1pt]  (v23) -- (v25);
    \draw[-, line width=1pt]  (v24) -- (v25);
    \draw[-, line width=1pt]  (v25) -- (v26);

    \path (v2) -- (v3) node [black, font=\Large, midway, sloped] {$\,\dots$};
    \path (v5) -- (v6) node [black, font=\Large, midway, sloped] {$\,\dots$};
    
    \path (v8) -- (v9) node [black, font=\Large, midway, sloped] {$\,\dots$};
    \path (v11) -- (v12) node [black, font=\Large, midway, sloped] {$\,\dots$};
    \path (v14) -- (v15) node [black, font=\Large, midway, sloped] {$\,\dots$};
    \path (v17) -- (v18) node [black, font=\Large, midway, sloped] {$\,\dots$};
    \path (v19) -- (v20) node [black, font=\Large, midway, sloped] {$\,\dots$};
    \path (v21) -- (v22) node [black, font=\Large, midway, sloped] {$\,\dots$};
    \path (v23) -- (v24) node [black, font=\Large, midway, sloped] {$\,\dots$};

    \path (9,0.5) -- (12,0.5) node [black, font=\Large, midway, sloped] {$\,\dots$};

    \draw [decorate,decoration={brace, amplitude=10pt},xshift=0pt,yshift=-1pt] (1.7,-0.1) -- (-0.2,-0.1) node [black,midway,yshift=-0.7cm] {};
    \node[] () at (0.75,-0.65) []{\large $d_{i,j}-2$};
    \draw [decorate,decoration={brace, amplitude=10pt},xshift=0pt,yshift=-1pt] (4.2,-0.1) -- (2.3,-0.1) node [black,midway,yshift=-0.7cm] {};
    \node[] () at (3.25,-0.65) []{\large $d_{i,j}-2$};
    
    \draw [decorate,decoration={brace, amplitude=10pt},xshift=0pt,yshift=-1pt] (4.2,-0.9) -- (-0.2,-0.9) node [black,midway,yshift=-0.7cm] {};
    \node[] () at (2,-1.5) []{\large $n^2$ stars};

    \draw [decorate,decoration={brace, amplitude=10pt},xshift=0pt,yshift=-1pt] (6.7,-0.1) -- (4.8,-0.1) node [black,midway,yshift=-0.7cm] {};
    \node[] () at (5.75,-0.65) []{\large $d_{i,j}-3$};
    \draw [decorate,decoration={brace, amplitude=10pt},xshift=0pt,yshift=-1pt] (9.2,-0.1) -- (7.3,-0.1) node [black,midway,yshift=-0.7cm] {};
    \node[] () at (8.25,-0.65) []{\large $d_{i,j}-3$};

    \draw [decorate,decoration={brace, amplitude=10pt},xshift=0pt,yshift=-1pt] (9.2,-0.9) -- (4.8,-0.9) node [black,midway,yshift=-0.7cm] {};
    \node[] () at (7,-1.5) []{\large $n^2$ stars};

    \draw [decorate,decoration={brace, amplitude=10pt},xshift=0pt,yshift=-1pt] (13.7,-0.1) -- (11.8,-0.1) node [black,midway,yshift=-0.7cm] {};
    \node[] () at (12.75,-0.65) []{\large $d_{i,j}-n^2$};
    \draw [decorate,decoration={brace, amplitude=10pt},xshift=0pt,yshift=-1pt] (16.2,-0.1) -- (14.3,-0.1) node [black,midway,yshift=-0.7cm] {};
    \node[] () at (15.25,-0.65) []{\large $d_{i,j}-n^2$};

    \draw [decorate,decoration={brace, amplitude=10pt},xshift=0pt,yshift=-1pt] (16.2,-0.9) -- (11.8,-0.9) node [black,midway,yshift=-0.7cm] {};
    \node[] () at (14,-1.5) []{\large $n^2+q$ stars};

    \draw [decorate,decoration={brace, amplitude=10pt},xshift=0pt,yshift=-1pt] (16.2,-1.8) -- (-0.2,-1.8) node [black,midway,yshift=-0.7cm] {};
    \node[] () at (8,-2.5) []{\large $2in^4-a'_j-1$ stars};

\end{tikzpicture}
}
\caption{The tree $T_{i,j}$ that is attached to the vertex $u_i$, where $j$ is such that $a_j\in \overline{L}(u_i)$, in the proof of Theorem~\ref{thm:edge-irr-w-hard}. The value of $q$ is such that after attaching $T_{i,j}$ to $u_i$ (and thus including the edge $u_{i,j}u_i$) we have $d(u_{i,j})=2in^4-a'_j$.}\label{fig:W-hard-fvs-edge}
\label{fig:Tij}
\end{figure}

First, for every $1\leq k\leq n^2-1$,
create $n^2$ 
copies of $S_{d_{i,j}-k}$ (the star on $d_{i,j}-k$ vertices) and $q$ additional copies of $S_{d_{i,j}-n^2+1}$
(the exact value of $q$ will be defined in what follows). Then, choose one leaf from each one of the above stars, and identify them into a single vertex denoted as $u_{i,j}$; the value of $q$ is such that $d(u_{i,j})=d_{i,j}=2in^4-a'_j$. 
Let $T_{i,j}$ be the resulting tree and let us say that $u_{i,j}$ is the root of $T_{i,j}$ (see Figure~\ref{fig:Tij}).  

Let us now describe the construction of $G$. For each vertex $u_i\in V(H)$ and for each $a_j\in \overline{L}(u_i)$, add the tree $T_{i,j}$ to $H$ and the edge $u_{i,j}u_i$. Then, for each vertex $u_i\in V(H)$, for any $j$ such that $u_{i,j}$ is a neighbour of $u_i$, add $p_i$ additional copies of the tree $T_{i,j}$, as well as the edges between $u_i$ and the roots of the additional trees, so that $d_G(u_i)=2in^4$. 
The resulting graph is $G$. Note that, for each vertex of $V(H)$, we are adding at most $\mathcal{O}(n)$ trees, each one containing at most $\mathcal{O}(n^{10})$ vertices. Thus, the construction of $G$ is achieved in polynomial time. 

\medskip 

\noindent\textbf{Reduction.} Assume first that $(H,L)$ is a yes-instance of \textsc{General Factor}, and let $S\subseteq E$ be such that $d_{H-S}(u)\in L(u)$ for all $u\in V(H)$. We claim that $S$ is also an edge-irregulator of $G$. By the construction of $G$, and since $S$ only contains edges from $H$, there are no two adjacent vertices in $G[V(G) \setminus V(H)]$ that have the same degree in $G-S$. Thus, it remains to check the pairs of adjacent vertices $x,y$ such that, either both $x$ and $y$ belong to $V(H)$, or, w.l.o.g., $x\in V(H)$ and $y\in V(G-H)$. For the first case, let $x=u_i$ and $y=u_{i'}$, for $1\leq i<i'\leq n$. Then, assuming that $d_{G-S}(u_i)=d_{G-S}(u_{i'})$, we get that $2in^4-p=2i'n^4-p'$, where $S$ contains $0\leq p< n^2$ and $0\leq p'< n^2$ edges incident to $u_i$ and $u_{i'}$ respectively. Thus, $2n^4(i-i')=p-p'$, a contradiction since $|p-p'| < n^2$ and $|2n^4(i-i')|\ge n^4$ (as $|i- i'|\ge 1$). For the second case, for every $i$, let $d_{G-S}(u_i)=2in^4-p$, where the set $S$ contains $0\leq p\leq n^2-1$ edges of $H$ incident to $u_i$. Also, by the construction of $G$ and since $S$ only contains edges from $H$, we have that for every $j$, $d_{G-S}(u_{i,j})=d_G(u_{i,j})=2in^4-a'_j$, where, recall, $a'_j=d_H(u_i)-a_j$ for $a_j\in \overline{L}(u_i)$ (see Figure~\ref{fig:W-hard-fvs-edge}). Assume now that there exist $i,j$ such that $d_{G-S}(u_i)=d_{G-S}(u_{i,j})$. Then, $2in^4-p=2in^4-d_H(u_i)+a_j$ and thus $d_H(u_i)-p=a_j$. But then $d_{H-S}(u_i)=a_j$, which is a contradiction since $a_j\in \overline{L}(u_i)$. Thus, $S$ is an edge-irregulator of $G$ and $|S|< n^2$ since $S$ only contains edges of $E(H)$.

For the reverse direction, assume that $\I_e(G)< n^2$ and let $S$ be an optimal edge-irregulator of $G$. We will show that $S$ is also such that $d_{H-S}(u_i)\in L(u_i)$, for every $i$. Let us first prove the following claim.

\begin{claimOC}\label{claim:no-edges-trees}
Let $S$ be an optimal edge-irregulator of $G$. Then either $S\subseteq E(H)$ or $|S|\ge n^2$.

\end{claimOC}
\begin{proofclaim}
    Assume there exist $i,j$ such that $|S\cap E_{i,j}|=x\geq 1$ and $x < n^2$, where $E_{i,j}$ is the set containing all the edges of the tree $T_{i,j}$ and the edge $u_iu_{i,j}$. Among those edges, there are $x_1\geq 0$ edges incident to $u_{i,j}$ and $x_2\geq 0$ edges incident to children of $u_{i,j}$ (but not to $u$), with $x_1+x_2=x< n^2$.
    
    Assume first that $x_1=0$. Then $x=x_2$ and there is no edge of $S\cap E_{i,j}$ that is incident to $u_{i,j}$. Then $d_{G-S}(u_{i,j})=d_G(u_{i,j})$ and observe that $d_{G}(u_{i,j})$ is strictly larger than that of any of its children (by the construction of $G$). It follows that $S\setminus E_{i,j}$ is also an edge-irregulator of $G$, contradicting the optimality of $S$. Thus $n^2>x_1\geq 1$. It then follows from the construction of $G$ that there exist at least $n^2$ children of $u_{i,j}$, denoted by $z_1,\dots,z_{n^2}$, such that $d_{G-S}(u_{i,j})=d_G(z_k)$, for every $1\leq k\leq n^2$. Since  $x < n^2$, there exists at least one $1\leq k\leq n^2$ such that $d_{G-S}(u_{i,j})=d_{G-S}(z_k)$, contradicting the fact that $S$ is an edge-irregulator. Thus either $x=0$ or $x\ge n^2$.
    \end{proofclaim}

It follows directly from Claim~\ref{claim:no-edges-trees} that $S$ contains only edges of $E(H)$. Assume that there exist $i,j$ such that $d_{H-S}(u_i)=a_j$ and $a_j\in \overline{L}(u_i)$. Then $d_{G-S}(u_i)=2in^4-a'_j$. Also, by the construction of $G$ and since $S\subseteq E(H)$, $u_i$ is adjacent to a vertex $u_{i,j}$ in $G-S$ and we have that $d_{G-S}(u_{i,j})=d_{G}(u_{i,j})=2in^4-a'_j$. This is contradicting the fact that $S$ is an edge-irregulator of $G$. Thus, for every $i,j$, we have that if $d_{H-S}(u_i)=a_j$, then $a_j\in L(u_i)$, which finishes our reduction.

Finally, if $H$ has vertex cover number $vc$, then, by Observation~\ref{obs:bounded-td-fvs}, 
we have that $G$ has treedepth and feedback vertex set $\mathcal{O}(vc)$. 
\end{proof}

We close this section by observing that the proof of Theorem~\ref{thm:fpt-vi} can be adapted for the case of edge-irregulators. Indeed, it suffices to replace the guessing of vertices and the variables defined on vertices, by guessing of edges and variables defined on the edges of the given graph. Finally, the definition of the sub-types is done through subgraphs produced only by deletion of edges. This leads us to the following:
\begin{corollary}
Given a graph $G$ with vertex integrity $k$, there exists an algorithm that computes $\I_e(G)$ in FPT-time parameterised by $k$.
\end{corollary}

\section{Conclusion}\label{sec:conclusion}
In this work we continued the study of the problem of finding optimal vertex-irregulators, and introduced the problem of finding optimal edge-irregulators. In the case of vertex-irregulators, our results are somewhat optimal, in the sense that we almost characterise which are the ``smallest'' graph-structural parameters that render this problem tractable. The only ``meaningful'' parameter whose behaviour remains unknown is the modular-width of the input graph. The parameterised behaviour of the case of edge-irregulators is also somewhat understood, but there are still some parameters for which the problem remains open. 
Another interesting direction is that of approximating optimal vertex or edge-irregulators. In particular it would be interesting to identify parameters for which either problem becomes approximable in FPT-time (recall that vertex-irregulators are not approximable within any decent factor in polynomial time~\cite{FMT22}). Finally, provided that the behaviour of edge-irregulators is better understood, we would also like to propose the problem of finding locally irregular minors, of maximum order, of a given graph $G$.

\acknowledgements
\label{sec:ack}
The authors would like to thank Dániel Marx for his important contribution towards rendering the proof of Theorem~\ref{thm:FPT-neighborhood diversity} more elegant.

\nocite{*}
\bibliographystyle{abbrvnat}
\bibliography{bibliography}
\label{sec:biblio}

\end{document}